\documentclass[a4paper,onecolumn,11pt,accepted=2023-06-04, allowfontchangeintitle]{quantumarticle}
\pdfoutput=1

\usepackage[utf8]{inputenc}
\usepackage{mathtools}
\usepackage{amsfonts}
\usepackage{mleftright}
\usepackage{enumerate}
\usepackage{amsthm}
\usepackage[pagebackref,colorlinks,citecolor=blue,bookmarks=true]{hyperref}
\usepackage[capitalize]{cleveref}
\usepackage[sort&compress]{natbib}
\setcitestyle{numbers,square,comma}
\usepackage{braket}
\usepackage{tikz}
\usetikzlibrary{arrows.meta, positioning, shapes.geometric}

\usepackage{cancel}

\theoremstyle{plain}
\newtheorem{theorem}{Theorem}[section]
\newtheorem{lemma}[theorem]{Lemma}
\newtheorem{corollary}[theorem]{Corollary}
\newtheorem{proposition}[theorem]{Proposition}

\newtheorem{observation}[theorem]{Observation}

\theoremstyle{definition}
\newtheorem{definition}[theorem]{Definition}

\newcommand{\Sp}{\text{Sp}}

\newcommand{\np}{\textsf{NP}}
\newcommand{\cnot}{\text{CNOT}}
\let\phi\varphi

\DeclarePairedDelimiter{\abs}{|}{|}
\DeclarePairedDelimiterX{\ketbra}[2]{|}{|}{%
	#1\rangle\!\langle#2%
}
\DeclareMathOperator*{\ex}{\mathbb{E}}
\DeclareMathOperator*{\pr}{\mathbb{P}}
\DeclareMathOperator{\trc}{Tr}
\DeclareMathOperator{\poly}{poly}

\newcommand{\equal}{=}

\begin{document}
	
	\title{Clifford Circuits can be Properly PAC Learned if and only if \texorpdfstring{$\textsf{RP} = \textsf{NP}$}{RP = NP}}
	
	\author{Daniel Liang}
	\affiliation{Department of Computer Science \\ University of Texas at Austin}
	\orcid{0000-0002-7418-0468}
	\email{dliang@cs.utexas.edu}
	\maketitle
	
	\begin{abstract}
		Given a dataset of input states, measurements, and probabilities, is it possible to efficiently predict the measurement probabilities associated with a quantum circuit? Recent work of \citet{2020Caro} studied the problem of PAC learning quantum circuits in an information theoretic sense, leaving open questions of computational efficiency. In particular, one candidate class of circuits for which an efficient learner might have been possible was that of Clifford circuits, since the corresponding set of states generated by such circuits, called stabilizer states, are known to be efficiently PAC learnable \citep{rocchetto2018stabiliser}. Here we provide a negative result, showing that proper learning of CNOT circuits with $1/\poly(n)$ error is hard for classical learners unless $\textsf{RP} = \np$, ruling out the possibility of strong learners under standard complexity theoretic assumptions. As the classical analogue and subset of Clifford circuits, this naturally leads to a hardness result for Clifford circuits as well. Additionally, we show that if $\textsf{RP} = \np$ then there would exist efficient proper learning algorithms for CNOT and Clifford circuits. By similar arguments, we also find that an efficient proper quantum learner for such circuits exists if and only if $\np \subseteq \textsf{RQP}$. We leave open the problem of hardness for improper learning or $\mathcal{O}(1)$ error to future work.
	\end{abstract}
	
	\section{Introduction}
	The goal of efficient learning of quantum states and the circuits that act on them, is to be able to predict the outcome of various measurements with some degree of accuracy. For example, given a quantum state $\rho$ and a two-outcome measurement $M$ can we predict the probability that the measurement accepts?
	
	Naively, one can try and learn everything there is to know about the system, a technique known as tomography with versions for quantum states \citep{10.1145/2897518.2897544, 10.1145/3055399.3055454, Haah_2017} and quantum processes \citep{doi:10.1080/09500349708231894, PhysRevLett.90.193601}. However, this requires exponential time in the number of qubits due to information theoretic reasons related to the exponential dimension of the system. This exponential bound remains when trying to find a state close in trace distance as shown by a combination of \citet{Flammia_2012} and Holevo's bound. To address this, one can choose to restrict the type of information one wanted to learn, which led to the ideas of shadow tomography~\citep{10.1145/3188745.3188802} and classical shadows~\citep{Huang2020}. By only needing to predict the value of $M$ observables $\{O_i\}$, one is able to use only a number of measurements that is polynomial in the number of qubits and polylogarithmic in $M$. In a similar vein, \citet{aaronson2007learnability} proposed the idea of PAC learning quantum states, which is the idea of learning relative to some distribution over measurements, but only being given samples from that distribution as well (see \cref{ssec:pac} for details). Aaronson was then able to give a generalization theorem for this problem, showing that if one could find a hypothesis $h$ that had small training error on $\mathcal{O}(n)$ samples that $h$ would also perform well on future samples. However, the problem of efficiently finding $h$ was left open.
	
	An alternative direction was to restrict the class of objects being learned on, but allow one to choose what kind of measurements are taken.
	\citet{montanaro2017learning} was able to learn stabilizer states and later \citet{PhysRevA.80.052314} learned an unknown Clifford circuit.
	\citet{Lai2022} built on these results in the case of actually recovering the circuit, as well as limited learning in the presence of a small amount of non-Clifford gates.
	Stabilizer states and Clifford circuits are of particular interest to the quantum information community because many quantum communication protocols and well known quantum query algorithms \citep{PhysRevLett.70.1895, PhysRevLett.69.2881, QKD, doi:10.1137/S0097539796298637, doi:10.1137/S0097539796300921, quant-ph/9705052} utilize these states and circuits.
	Gottesman and Knill \citep{https://doi.org/10.48550/arxiv.quant-ph/9807006} (with improvements by \citet{aaronson2004improved}) were able to give an efficient classical simulation of these objects, showing these class of objects to seemingly be much simpler than the set of all quantum states or circuits.
	Combined with the fact that stabilizer states are good approximations to Haar random states \citep{PhysRevA.96.062336, https://doi.org/10.48550/arxiv.1510.02767}, we get a set of circuits and states that are highly quantum with many interesting uses, but still have enough exploitable structure to be classically simulable, making them a prime candidate for learning.
	
	\citet{rocchetto2018stabiliser} was able to combine the ideas of PAC learning with the structure provided by restricting to stabilizer states to give an efficient PAC algorithm for learning stabilizer states. \citet{2020Caro} extended the ideas of PAC learning to quantum circuits, giving an analogous generalization theorem to \citet{aaronson2007learnability}. As with \citet{aaronson2007learnability}, the problem of efficiently finding such a good hypothesis was left open. A natural follow-up was whether or not Clifford circuits could be efficiently PAC learned in an analogous way to stabilizer states. Here, we are given inputs of the form $\mleft(\rho, \frac{I^{\otimes n}+P}{2}\mright)$ for some stabilizer state $\rho$ and Pauli matrix $P$, with labels $\trc\mleft[\frac{I^{\otimes n}+P}{2}C\rho C^\dagger\mright]$ corresponding to an unknown Clifford circuit $C$ and asked to predict future labels.
	It is worth noting that we have slightly altered the definition of PAC learning a quantum circuit from that of \citet{2020Caro} to a setting we find more comparable to Aaronson's original PAC learning result for quantum states \citep{aaronson2007learnability}. In the setting introduced by Caro and Datta, the measurements were limited to being rank 1 projectors with product structure, rather than the rank $2^{n-1}$ projectors we use in our proof.
	
	When one attempts to create a PAC learning algorithm a natural first step is to try an elimination method, i.e., eliminating options that don't match the given training data and then outputting some option that does match the data well.
	Such algorithms are known as \textit{proper} learning algorithms (see \cref{ssec:pac} for more details) and were the only kind of learning algorithms considered when the idea of PAC learning was first introduced by \citet{valiant1984theory}. And while the learning theory community now considers things like improper learning algorithms, the original proper learning algorithms generally remain the most natural class of learning algorithms to consider first. We note for instance that \citet{rocchetto2018stabiliser} is a proper learning algorithm, as well as learning algorithms for parities and other well known learning problems \citep{Klivans2005}
	
	To that extent, we show in this paper that an efficient proper learner for Clifford circuits that achieves $1/\poly(n)$ error exists if and only if $\textsf{RP} = \np$, effectively ruling out ``straight-forward'' learning algorithms for Clifford circuits.
	More generally, these results apply to any proper learner that achieves arbitrary error $(\epsilon, \delta)$ with runtime $\poly(n, \epsilon^{-1}, \delta^{-1})$, which is known as a \textit{strong} learner.
	Furthermore, this is true even just for a learner of a subset of Clifford circuits called CNOT circuits.
	This subset essentially restricts to the set of Clifford circuits that map computational basis states to other computational basis states and these circuits are highly related to the complexity class $\textsf{L}$ (see \cref{ssec:cnot_circuit}). 
	We leave open the problem of showing $\mathcal{O}(1)$ hardness for proper learners using complexity theoretic means, such as in \citet{venkatesan_2009}.

	One can also imagine that the learning algorithm has access to a quantum computer. Since there exists problems like factoring \citep{doi:10.1137/S0036144598347011} for which we have an efficient quantum algorithm but not an efficient classical algorithm, this learner may be able to efficiently learn more expressive concept classes. We also give results for this setting by relating $\np$ to $\textsf{RQP}$, the quantum analogue of $\textsf{RP}$. We now informally state our main theorems regarding CNOT and Clifford circuits.
	
	\begin{theorem}\label{thm:main}
		There exists an efficient randomized proper PAC learner for CNOT circuits if and only if $\textsf{RP} = \np$. Furthermore, an efficient quantum proper PAC learner for CNOT circuits exists if and only if $\np \subseteq \textsf{RQP}$.
	\end{theorem}
	
	\begin{corollary}\label{cor:main}
		There exists an efficient randomized proper PAC learner for Clifford circuits if and only if $\textsf{RP} = \np$. Furthermore, an efficient quantum proper PAC learner for Clifford circuits exists if and only if $\np \subseteq \textsf{RQP}$.
	\end{corollary}

	The proofs of these main results starts by realizing that finding a CNOT circuit with zero training error requires finding a full rank matrix in an affine subspace of matrices under matrix addition (so as to differentiate from a coset of a matrix group using matrix multiplication). This is known as the \textsc{NonSingularity} problem \citep{BUSS1999572} and is \np-complete. While this may seem like a backwards reduction, it turns out that the set of matrix affine subspaces used to show that \textsc{NonSingularity} can solve 3SAT are a subset of the ones needed to learn CNOT circuits with zero training error. Thus, there exist a set of samples such that a CNOT circuit with zero training error exists if and only if the SAT instance is satisfiable. Finding such a CNOT circuit is what is known as the search version of the \textit{consistency problem} and in turn the decision version of the consistency problem is also \np-complete.
	
	To show that an efficient proper learner for CNOT circuits implies $\textsf{RP} = \np$, we follow the same proof structure as similar results for \np-hardness of the consistency problem for 2-clause CNF, 3-DNF, or the intersection of two halfspaces \citep{blum2015dp, BLUM1992117, Haghtalab2020}. First, let $S$ be some sample from the decision version of the consistency problem for CNOT circuits. Using the uniform distribution over each element in $S$, we will sample every element of $S$ with high probability given enough queries. Since $S$ contains at most a polynomial number of samples, we are able to show that an efficient learner with arbitrary $1/\poly(n)$ error would necessarily also solve the consistency problem with high enough probability to create a solution in $\textsf{RP}$.
	
	Completing the proof in the other direction, if $\textsf{RP} = \np$ we utilize search-to-decision reductions for $\np$-complete problems to get an efficient algorithm for the search problem of minimizing training error. We can treat this search algorithm as our means of generating a hypothesis circuit $C$ with low training error. By the generalization theorem provided by \citet{2020Caro}, assuming we have enough samples, this $C$ will properly generalize and have low true error, thus completing the proof. The quantum forms of the proof essentially come for free by replacing $\textsf{RP}$ with $\textsf{RQP}$ everywhere and using learners capable of doing quantum computation.

	\subsection{Related Work}
	We also note that we are dealing with the problem of classically PAC learning a classical function (i.e., classical labels) derived from a quantum system. This is as opposed to quantum PAC learning of a classical function as in \citet{arunachalam2017guest, arunachalam2018optimal, arunachalam2020quantum, arunachalam2021private} where instead of a distribution over samples we receive access to copies of a quantum state. This state results in the same distribution classically when measured in the computational basis but can be measured in other basis to get different results. There is also the attempt to directly learn a quantum process with quantum labels, as in \citet{chung_et_al_2021, Caro_2021}. Here, they do not choose to measure the output state, and have samples of the form $(\rho, \mathcal{M}(\rho))$ for quantum process $\mathcal{M}$. Other related quantum learning works, some of which are outside the PAC model, include \citet{yoganathan2019condition, PhysRevA.80.052314, cheng2015learnability}.
	
	\section{Preliminaries}
	
	\subsection{Quantum States and Circuits}
	A quantum state $\rho$ on $n$ qubits is a $2^n \times 2^n$ PSD matrix with trace $1$. If the matrix is rank $1$ then we refer to $\rho$ being a \textsc{pure state}, since it can be decomposed as $\rho = \ketbra{\psi}{\psi}$ where $\ket{\psi}$ is a $2^n$-dimensional column vector with norm 1 and $\bra{\psi}$ is its complex conjugate. A two-outcome measurement $E$ is then a projector such that $E^2 = E$ such that the probability of a `1' outcome is $\trc\mleft[E \rho\mright]$ and the probability of a `0' outcome is $1-\trc\mleft[E \rho\mright]$, leaving the expectation value as simply $\trc\mleft[E \rho\mright]$.
	
	A quantum process is how one evolves a quantum state, and therefore it must preserve the trace $1$ and the PSD condition. We will be primarily interested in quantum circuits, which are the subset of quantum processes that map pure states only to other pure states. These are constrained to be unitary operations, such that after acting on $\rho$ with the circuit $C$, the state that we are left with is $C \rho C^\dagger$ where $C^\dagger$ is the complex conjugate of $C$.
	
	\subsection{Paulis and Stabilizer States/Groups}
	
	We will start by giving the following matrices, known as the \textsc{Pauli matrices}.
	
	\[
	I = \begin{pmatrix}1 & 0\\0 & 1\end{pmatrix}
	\,\,\,\,X = \begin{pmatrix}0 & 1\\1 & 0\end{pmatrix}
	\,\,\,\,Y = \begin{pmatrix}0 & -i\\i & 0\end{pmatrix}
	\,\,\,\,Z = \begin{pmatrix}1 & 0\\0 & -1\end{pmatrix}
	\]
	
	Noting that these are all unitaries that act on a single qubit, we can generalize to $n$ qubits.
	
	\begin{definition}
		Let $\mathcal{P}_n = \{\pm 1, \pm i\} \times \{I, X, Y, Z\}^{\otimes n}$ be the matrix group consisting all $n$-qubit Paulis with phase $\pm 1$ or $\pm i$.
	\end{definition}
	
	We'll also introduce some shorthand notation:
	
	\begin{definition}
		Let $X_i$ and $Z_i$ be the Pauli acting only on the $i$-th qubit with $X$ or $Z$ respectively and the identity matrix on all other qubits.
	\end{definition}
	
	\begin{definition}\label{def:power_z}
		For $v \in \{0, 1\}^n$, let $X^v = \prod_{i=1}^n X_i^{v_i}$ and $Z^v = \prod_{i=1}^n Z_i^{v_i}$. 
	\end{definition}
	
	Note that $Z^{v} \cdot Z^{w} = Z^{v + w}$, assuming the dimensions of $v$ and $w$ match. It is easy to see that $v \neq w$ also implies that $Z^v \neq Z^w$.

	A \textsc{stabilizer state} $\rho$ is any state that can be written as $\frac{1}{2^n}\sum_{g \in G} g$, where $G$ is an abelian subgroup $G \subset \mathcal{P}_n\setminus \{-I^{\otimes n}\}$ without the negative identity. $G$ is known as the \textsc{stabilizer group} of $\rho$. As it turns out, if $G$ is of order $2^n$ then $\rho$ will be a pure state. This leads to the alternative (and more popular definition) where $\rho = \ketbra{\psi}{\psi}$, is the unique state that is stabilized by $G$. That is, for all $g \in G$, $g \ket{\psi} = \ket{\psi}$. This definition shows why $-I^{\otimes n}$ isn't allowed to be in $G$, since $-I^{\otimes n}$ stabilizes nothing. It also shows why one must restrict the entries of $G$ to only have real phase.
	
	\begin{proposition}
		Any abelian subgroup of $G \subseteq \mathcal{P}_n\setminus \{-I^{\otimes n}\}$ cannot contain any Paulis with an imaginary phase.
	\end{proposition}
	\begin{proof}
		Given a Pauli with an imaginary phase, it's square would be equal to $-I^{\otimes n}$, making the group not closed. This is a contradiction.
	\end{proof}
	
	One of the reasons stabilizer states are so important is this bijection between the stabilizer group of a stabilizer state and the state itself; by simply knowing the generators of the group, one can easily reconstruct the state. And since there are at most $n$ generators, if one can efficiently write down the generators themselves then there is a polynomial size representation of a stabilizer state. We now show how one can write down any member of a stabilizer group as follows. Given, $P \in \mathcal{P}_n$ with real phase such that $P = \pm \bigotimes_i P_i$, define a function $N: \mathcal{P} \rightarrow \{0, 1\}^2$ for each qubit $N(I) = 00$, $N(X) = 10$, $N(Z) =11$, and concatenate to make $N(P) = (N(P_1), N(P_2), \dots, N(P_n))$. Additionally, have an extra bit for the sign for whether the sign is $-1$ or $1$. This results in a $2n+1$ bit string for each generator, so writing down a stabilizer state requires only $\mathcal{O}\mleft(n^2\mright)$ bits to write down classically.

	\subsection{Clifford Circuits}

	Informally, a Clifford circuit maps stabilizer states to other stabilizer states.
	
	\begin{definition}
		A \textsc{Clifford circuit} is a unitary $U$ such that $U\mathcal{P}_nU^\dagger = \mathcal{P}_n$, while ignoring global phase on the unitary. More formally, consider the normalizer $\mathcal{N}(\mathcal{P}_n) = \{U \in U(2^n) \mid U\mathcal{P}_n U^\dagger = \mathcal{P}_n\}$, and let $\mathcal{C}_n = \mathcal{N}(\mathcal{P}_n)/U(1)$ be the \textsc{Clifford group}.
	\end{definition}
	
	Like stabilizer states, generators are an important part of how we deal with Clifford circuits. How a given Clifford circuit $U$ acts on the generators of the Pauli matrices completely characterizes the unitary $U$ \citep{PhysRevA.80.052314}. To borrow the notation of \citet{2014Koenig}, this relationship can be efficiently described via:
	
	\begin{align}\label{eq:clifford_gen}
		UX_j U^\dagger = (-1)^{p_j}\prod_{i=1}^n X_i^{\alpha_{ij}} Z_i^{\beta_{ij}}
		\hspace{2em}
		UZ_j U^\dagger = (-1)^{q_j}\prod_{i=1}^n X_i^{\gamma_{ij}} Z_i^{\theta_{ij}}
	\end{align} 
	where $p_i$, $q_i$, $\alpha_{ij}$, $\beta_{ij}$, $\gamma_{ij}$, and $\theta_{ij}$ are all $\{0, 1\}$ values.
	It will sometimes be useful to view $\alpha_{ij}$, $\beta_{ij}$, $\gamma_{ij}$, and $\theta_{ij}$ as the $n \times n$ boolean matrices $A$, $B$, $\Gamma$, and $\Theta$ respectively. This gives us a simple upper-bound on the number of Clifford circuits.
	
	\begin{proposition}\label{prop:clifford_circuit_number}
		There are at most $2^{\mathcal{O}(n^2)}$ Clifford circuits.
	\end{proposition}
	\begin{proof}
		The total number of bits we use to represent $p$, $q$, $A$, $B$, $\Gamma$, and $\Delta$ is $4n^2 + 2n = \mathcal{O}(n^2)$. There can then be at most $2^{\mathcal{O}(n^2)}$ Clifford circuits.
	\end{proof}
	
	\sloppy
	However, because commutation relations are preserved, not all possible values of $\alpha, \beta, \gamma, \theta$ are allowed (the $p$ and $q$ values can be arbitrary).
	This leads us to the idea of symplectic matrices.
	We note that a Clifford circuit can be encoded as a $(2n + 1) \times 2n$ boolean matrix $S$ where column $2j-1$ is equal to $(\alpha_{1j}, \beta_{1j}, \cdots, \alpha_{nj}, \beta_{nj}, p_j)$ and column $2j$ is equal to $(\gamma_{1j}, \theta_{1j},\cdots, \gamma_{nj}, \theta_{nj}, q_j)$.
	\fussy
	We will call this the \textit{full encoding} of the Clifford circuit.

	\begin{definition}
		A symplectic matrix over $\mathbb{F}_2^{2n}$ is a $2n \times 2n$ matrix $S$ with entries in $\mathbb{F}_2$ such that
		\begin{align}\label{eq:symplectic_matrix}
			S^T \Lambda(n)S = \Lambda(n) \equiv \bigoplus_{i=1}^n \begin{pmatrix} 0 & 1 \\ 1 & 0 \end{pmatrix}.
		\end{align} These matrices form the symplectic group $\Sp(2n, \mathbb{F}_2)$.
	\end{definition}

	The symplectic matrices preserve the symplectic inner product $\omega(v, w) = v^T \Lambda(n) w$ on $\mathbb{F}_2^{2n}$. It turns out that if we consider the submatrix defined by the first $2n$ rows of our full encoding $S$, a necessary and sufficient condition to preserve the commutation relations of the generators is for this submatrix to be symplectic, as $\{X_i\} \cup \{Z_i\}$ form what is known a symplectic basis.
	Formally, $\mathcal{C}_n / \mathcal{P}_n \cong \Sp(2n, \mathbb{F}_2)$.
	
	\subsection{CNOT circuits and \texorpdfstring{$\oplus\textsf{L}$}{\ensuremath{⊕}L}}\label{ssec:cnot_circuit}
	
	It is a well known fact that every Clifford circuit can be generated using only $H$, $P$, and CNOT gates as defined below:
	
	\begin{align*}
		H = \frac{1}{\sqrt{2}}\begin{pmatrix}1 & 1\\1 & -1\end{pmatrix}
		\,\,\,\,\,\, P = \begin{pmatrix}1 & 0\\0 & i\end{pmatrix}
		\,\,\,\,\,\, \cnot = \begin{pmatrix}1 & 0 & 0 & 0\\0 & 1 & 0 & 0\\0 & 0 & 0 & 1\\0 & 0 & 1 & 0\end{pmatrix}
	\end{align*}
	
	We note that $X = HP^2H$. If we restrict to the subset of circuits that are generated by only $X$ and CNOT, we get what are known as CNOT circuits \citep{aaronson2004improved}, which are a clear subset of Clifford circuits.
	
	\begin{definition}[\citet{aaronson2004improved}]
		The complexity class $\oplus\textsf{L}$ is the class of problems that reduce to simulating a polynomial-size CNOT circuit.
	\end{definition}
	
	\noindent A perhaps more familiar definition for complexity theorists is the class of problems that are solvable by a nondeterministic logarithmic-space Turing machine that accepts if and only if the total number of accepting paths is odd.
	
	Let us now consider the set of all Clifford circuit that map computational basis states to other computational basis states, thereby stabilizing the subgroup $\{\pm 1\} \times \{I, Z\}^{\otimes n}$. Very briefly, we will call these \textit{classical Clifford circuits} as we will now prove that they are largely equivalent to CNOT circuits. The following lemmas will be useful.
	
	\begin{lemma}\label{lemma:theta_full_rank}
		Let $\Theta$ be the matrix form of the $\theta_{ij}$ from \cref{eq:clifford_gen}. Any CNOT circuit $C$ must have $\Theta$ be full rank.
	\end{lemma}
	\begin{proof}
		Let us first consider what happens to a computational basis state when acted upon by $C$ and let $S$ be the full encoding of $C$. Referencing \cref{eq:clifford_gen}, the $\gamma_{ij}$ must be $0$ for all $i$ and $j$. Since every member of $\Sp(2n, \mathbb{F}_2)$ is full rank, the even columns of $S$ must be as well. Since the $\gamma$ terms are all zero, the even columns of $S$ are full rank if and only if $\Theta$ is full rank.
	\end{proof}

	\begin{lemma}\label{lemma:theta_rowsum}
		Let $\Theta$ be the matrix form of the $\theta_{ij}$ from \cref{eq:clifford_gen} for some Clifford circuit $C$. If $\Theta$ is full rank then there exists a CNOT circuit with the same $\Theta$.
	\end{lemma}
	\begin{proof}
		One can verify that the $\Theta$ matrix of the circuit that does nothing, which is a valid CNOT circuit as well, is the identity matrix.
		We note that a CNOT from qubit $i$ to qubit $j$ performs the rowsum operation of adding row $j$ to row $i$ of $\Theta$. Thus it is possible to efficiently construct a circuit with matching $\Theta$ using rowsum operations via CNOT gates.
	\end{proof}
	
	We can now prove our desired goal leveraging these two lemmas.
	
	\begin{proposition}\label{prop:sim_cnot}
		Let $C$ be an arbitrary classical Clifford circuit. It can be efficiently generated using solely $X$, $Z$, and CNOT gates. Moreover, it's effect on the computational basis states can be entirely simulated using only $X$ and CNOT.
	\end{proposition}
	\begin{proof}
		Let us first consider what happens to a computational basis state when acted upon by $C$. Referencing \cref{eq:clifford_gen}, the $\gamma_{ij}$ must be $0$ for all $i$ and $j$, and we will essentially ignore $\alpha_{ij}$, $\beta_{ij}$, and $p_j$ for now leaving us with $\theta_{ij}$ and $q_{j}$. By \cref{lemma:theta_full_rank}, $\Theta$ must full rank. By \cref{lemma:theta_rowsum}, there exists a CNOT matrix that achieves the same $\Theta$ as well. To get a matching $q_j$, one can simply apply an X gate at the beginning of each qubit that has $q_j = 1$, since $XZX = -Z$, and the following CNOT gates will not itself introduce any negative phases. From here, we have already proved the \textit{moreover} statement.
		
		To prove the full result, we return to the $\alpha_{ij}$ and $\beta_{ij}$. We will show that there exists a single unique solution. Similar to $\Theta$ we will define the corresponding $n \times n$ matrices $A$ and $B$ for the $\alpha_{ij}$ and $\beta_{ij}$ respectively. Based on \cref{eq:symplectic_matrix}, to form a symplectic basis we find that $A^T \Theta = I$ and $A^T B = 0$, since $\gamma_{ij} = 0$. Clearly $A^T = \Theta^{-1}$, which is guaranteed to exist, and $B = 0$ since $A$ will also be full rank. To match the $p_j$ values we simply place $Z$ gates in front of the qubits where $p_j = 1$, similar to the $X$ gates for $q_j$.
	\end{proof}
	
	This means we do not lose any kinds of interactions by only considering CNOT circuits, since the only differentiating factors (i.e., the $Z$ gates) do not actually affect the outcome when fed with a computational basis state.
	As such, all given results will be given in terms of simply CNOT circuits.
	
	\subsection{PAC Learning}\label{ssec:pac}
	The goal of PAC learning is to learn a function relative to a certain distribution of inputs, rather than in an absolute sense. Let's say we want to learn an arbitrary $f$ from some concept class $\mathcal{C}$. If a hypothesis function $h$ matches the true function $f$ on many of the high probability inputs, then we can say that we have \textit{approximately learned} $f$. If we can do this with high probability for arbitrary $f$, then we \textit{probably} approximately (PAC) learned $\mathcal{C}$.
	
	\begin{definition}
		Let $\Omega$ be some domain of inputs and let $\mathcal{C}$ be a set of functions $f: \Omega \rightarrow [0, 1]$. We say that $\mathcal{C}$ is $(\epsilon, \delta)$-PAC-learnable if there exists a learner that, when given samples of the form $(x, f(x))$ for $x \sim \mathcal{D}$ for arbitrary $f$ and unknown distribution $\mathcal{D}$, outputs with probability at least $1-\delta$, over both the samples and the learning algorithm, a hypothesis $h$ with error\footnote{Since this problem is a regression problem rather than binary classification, we claim that this squared-loss is a more natural notion of error than the kind used in \citet{aaronson2007learnability}, \citet{2020Caro}, and other older papers \citep{anthony_bartlett_2000}. It also allows us to express PAC learning using only $2$ error parameters, rather than $3$. We can also still recover some form of the $3$ parameter bound using $\Pr_{x \sim \mathcal{D}}[|f(x) - h(x)| > t] \leq \frac{1}{t^2}\ex_{x \sim \mathcal{D}}\mleft[\mleft(f(x) - h(x)\mright)^2\mright]$ via Markov's inequality.} satisfying
		\[
		\ex_{x \sim \mathcal{D}}\mleft[\mleft(f(x) - h(x)\mright)^2\mright] \leq \epsilon.
		\]
		
		The number of samples used is referred to as the sample complexity, and we refer to the learner being \textit{efficient} if it can find such an $h$ in time $\poly(n, \epsilon^{-1}, \delta^{-1})$ for arbitrary $\epsilon$ and $\delta$.
	\end{definition}
	
	From here, one can define two types of learning, based on where $h$ comes from. If $h$ is allowed to be any function that meets the PAC constraints, we refer to this as \textsc{improper learning}. If instead $h \in \mathcal{C}$, we get what is known as \textsc{proper learning}, which will be the focus of this paper. With proper learning, we can then begin to talk about the consistency problem formally.
	
	\begin{definition}
		\sloppy Let $S$ be a set of labeled samples such that $|S| < s$. Let $\textsc{ConsistentSearch}(\mathcal{C}, s)$ be the problem of finding a function $h \in \mathcal{C}$ that is \textit{consistent} with all of $S$ (i.e., for all $(x, f(x)) \in S$, $f(x) = h(x)$) if such an $h$ exists, otherwise reject.
	\end{definition}

	\subsubsection{Generalization}
	
	Intuitively, given a set of samples the best one can really hope to do is find such an $h$ that gets zero training error and hope that the true error for $h$ is also low.
	This leads to the idea of \textit{generalization}, which aims to show that doing well on a large enough set of training data (i.e., the consistency problem) allows one to give the PAC guarantee as well with high probability.
	In terms of computational efficiency, this effectively reduces the problem of proper learning to the consistency problem, or an approximation of the consistency problem.
	The most common approach to guarantee generalization is to bound the ``expressiveness'' of the concept class, such as with the VC-dimension \citep{Blumer1989}.
	Since VC-dimension is defined for $\{0, 1\}$ labels, we will now give a generalization of VC-dimension in the regression setting.
	
	\begin{definition}
		Let $\Omega$ be some domain of inputs, $\eta > 0$, and let $\mathcal{C}$ be a set of functions $f: \Omega \rightarrow [0, 1]$. We say that a set of inputs $\{x_1, x_2, \cdots, x_m\} \subseteq \Omega$ is $\eta$-\textit{fat-shattered} by $\mathcal{C}$ if there exists a set $y_1, y_2, \cdots, y_m \in [0, 1]$ such that for any vector $b = \{\pm 1\}^m$ there is an $f_C \in \mathcal{C}$ that satisfies $b_i \cdot \mleft(f_C(x_i) - y_i \mright) \geq \eta$.
	\end{definition}
	
	\begin{definition}\label{def:fatshatterdim}
		The $\eta$-\textit{fat-shattering dimension} of a concept class $\mathcal{C}$ is the size of the largest set of inputs that is fat-shattered by $\mathcal{C}$. We denote this as $fat_{\mathcal{C}}(\eta)$.
	\end{definition}

	At a high-level, the $\eta$ parameter provides a buffer such that each $f_C(x_i)$ is robustly bounded away from $y_i$ by $\eta$ in the appropriate direction. We now give a result saying that bounded fat-shattering dimension implies generalization from the training data to the actual learning task.
	
	\begin{theorem}[\citet{anthony_bartlett_2000} Corollary 3.3]\label{thm:fat_to_samp}
		Let $\mathcal{C}$ be a concept class from $\Omega$ to $[0, 1]$ and let $\mathcal{D}$ be some distribution on $\Omega$. Let $\delta, \epsilon, \alpha, \beta \in (0, 1)$ be parameters such that $\beta > \alpha$. Furthermore, let $\{x_1, x_2, \cdots, x_m\}$ and $\{y_1, y_2, \cdots, y_m\}$ be a set of $m$ samples drawn i.i.d. from some distribution $\mathcal{D}$ where $y_i = f(x_i)$ for some $f \in \mathcal{C}$. If $h \in \mathcal{C}$ satisfies $\abs{h(x_i) - y_i} \leq \alpha$
		for all $1 \leq i \leq m$ then
		\[
			m = \mathcal{O}\mleft(\frac{1}{\epsilon}\mleft(fat_\mathcal{C}\mleft(\frac{\beta - \alpha}{8}\mright)\log^2\mleft(\frac{fat_\mathcal{C}\mleft(\frac{\beta - \alpha}{8}\mright)}{(\beta - \alpha)\epsilon}\mright) + \log \frac{1}{\delta}\mright)\mright)
		\]
		number of samples suffices to achieve
		\[
			\ex_{x \sim \mathcal{D}}\mleft[\mleft(h(x) - f(x)\mright)^2\mright] \leq (1-\epsilon)\beta^2 + \epsilon
		\]
		with probability at least $1- \delta$ over the samples.
	\end{theorem}

	The following folklore bound on fat-shattering dimension is very loose, but still sufficient for our purposes of complexity-theoretic hardness in \cref{sec:pac_np_hard}.
	
	\begin{lemma}[Folklore]\label{lem:fat_dim_finite}
		Given a concept class $\mathcal{C}$ such that $\abs{\mathcal{C}}$ is finite, then for all $\eta > 0$, $fat_\mathcal{C}(\eta) \leq \log_2 \abs{\mathcal{C}}$.
	\end{lemma}
	\begin{proof}
		\sloppy
		Assume for the sake of contradiction that $\mathcal{C}$ $\eta$-fat-shatters the set of points $\{x_1, x_2, \cdots, x_m\}$ for $m > \log_2 \abs{\mathcal{C}}$.
		Then $\mathcal{C}$ must be able to properly match each $b \in \{\pm 1\}^m$.
		There are $2^m > \abs{\mathcal{C}}$ possible $b$ vectors, and only one $f \in \mathcal{C}$ can be used per shattering attempt, since no $f$ can ever satisfy two different $b$ vectors.
		This is a contradiction since we don't have enough $f \in \mathcal{C}$ to go around to satisfy every $b$ vector.
	\end{proof}

	\subsubsection{Decision Problems}
	
	One can also define the decision version of the consistency problem, which is deciding if there even exists an $h \in \mathcal{C}$ that is consistent with all of $S$. We show that the existence of efficient learning algorithms can imply efficient one-sided error algorithms for the decision version of the consistency problem.
	
	\begin{definition}
		Let $\textsc{ConsistentDecide}(\mathcal{C}, s)$ be decision version of the consistency problem for $\mathcal{C}$ using at most $s$ samples.
	\end{definition}
	
	\begin{proposition}\label{prop:decision2learning}
		An efficient randomized $\mleft(\epsilon < \frac{\alpha^2}{s^2}, \delta < \frac{1}{2} + \frac{1}{2s}\mright)$\footnote{We abuse notation to signify that $\epsilon$ is a value less than $\frac{\alpha^2}{s^2}$ and likewise for $\delta < \frac{1}{2s}$.} proper learning algorithm implies $\textsc{ConsistentDecide}(\mathcal{C}, s) \in \textsf{RP}$ where $\alpha = \inf_{x \in \Omega, f(x) \neq g(x)} \mleft|f(x) - g(x)\mright|$ is the minimum non-zero error any hypothesis function can make on a single input.
	\end{proposition}
	\begin{proof}
		For every set of samples $S$ such that $|S| \leq s$, we can define the $\mathcal{D}_S$ to be the uniform distribution over all $x \in \chi$ such that $(x, f(x)) \in S$. By coupon collector, if we draw $\mathcal{O}(s \log s)$ many samples then with probability at least $1 - \frac{1}{s}$ we will have drawn every item from $S$. Now imagine that there exists some hypothesis $h \in \mathcal{C}$ that is not consistent with $S$. Then our error must be at least $\frac{\alpha^2}{s^2}$ by the definition of $\alpha$.
		
		Now assume we have some efficient randomized $(\epsilon, \delta)$ proper learning algorithm for $\epsilon < \frac{\alpha^2}{s^2}$ and $\delta < \frac{1}{2} + \frac{1}{2s}$. When running the learner on an arbitrary $\mathcal{D}_S$, it will see samples $S$ with probability at least $1-\frac{1}{s}$. To get error less than $\frac{\alpha^2}{s^2}$ the learner must then be able to solve the search version consistency problem with probability $p$ such that $\frac{1}{s} + (1-\frac{1}{s})p \geq 1-\delta$. Solving for $p$ we find $p \geq \frac{1}{2}$ on accepting instances.
		
		This gives rise to the following algorithm in $\textsf{RP}$ for solving $\textsc{ConsistentDecide}(\mathcal{C}, s)$. Given samples $S$ with $|S| \leq s$, we can run our learning algorithm and pretend that $S$ is what we sampled from $\mathcal{D}_S$ to get hypothesis $h$. If $h$ is consistent with $S$ then accept, otherwise reject. On an accepting instance $h$ will be consistent with probability at least $\frac{1}{2}$ while on rejecting instances it will never be consistent so the algorithm will always reject.
	\end{proof}
	
	Informally, if there exists enough structure on the concept class, it can be possible to go the other way and show that an efficient algorithm for $\textsc{ConsistentDecide}(\mathcal{C}, s)$ implies an efficient proper learner for $\mathcal{C}$. Namely, if a search-to-decision reduction exists for the consistency problem on $\mathcal{C}$ and $fat_\mathcal{C}$ is finite then we can also expect to show that an efficient algorithm for the decision problem would imply an efficient proper learner for $\mathcal{C}$. Of particular interest are \np-complete problems, which always admit search-to-decision reductions \citep{katz2011}. We can now give a formal proof of this commonly used technique to show proper PAC learning if $\mathsf{RP} = \np$.
	
	\begin{lemma}\label{lem:learning2decision}
		Let $\mathcal{C}$ be a concept class and let
		\[
		m = \Theta\mleft(\frac{1}{\epsilon}\mleft(fat_\mathcal{C}\mleft(\frac{\beta}{8}\mright)\log^2\mleft(\frac{fat_\mathcal{C}\mleft(\frac{\beta}{8}\mright)}{\beta\epsilon}\mright) + \log \frac{1}{\delta}\mright)\mright)
		\]
		be the parameter from \cref{thm:fat_to_samp} with $\alpha = 0$.
		If $\textsc{ConsistentDecide}(\mathcal{C}, s)$ for $s > m$ is \np-complete and $\mathsf{RP} = \np$ then for $c = \mathcal{O}(\poly(n))$ there exists an efficient algorithm to $\mleft((1-\epsilon)\beta^2 + \epsilon, \delta + \frac{1}{2^c}\mright)$ proper learn $\mathcal{C}$.
	\end{lemma}
	\begin{proof}
		\sloppy
		Because search-to-decision reductions exist for all \np-complete problems \citep{katz2011}, a zero-error oracle for $\textsc{ConsistentDecide}(\mathcal{C}, s)$ can be used to efficiently solve $\textsc{ConsistentSearch}(\mathcal{C}, s)$. Let us run our algorithm for $\textsc{ConsistentSearch}(\mathcal{C}, s)$ on a sample $S$ such that $s \geq |S| \geq m$. We now have an $h \in \mathcal{C}$ such that
		\[
		\mleft|h(x_i) - f(x_i) \mright| = 0 \,\,\,\,\, \forall1 \leq i \leq k
		\]
		and so by \cref{thm:fat_to_samp} 
		\[
		\ex_{x \sim \mathcal{D}}\mleft[\mleft(h(x) - f(x)\mright)^2\mright] \leq (1 - \epsilon)\beta^2 + \epsilon
		\]
		with probability at most $\delta$ over the samples.
		
		Finally, let $\gamma$ be the number of calls to $\textsc{ConsistentDecide}(\mathcal{C}, s)$ used in the search-to-decision reduction. In order for the reduction to be efficient, $\gamma = \mathcal{O}(\poly(n))$. Since $\textsc{ConsistentDecide}(\mathcal{C}, s)$ is in $\np$ and therefore \textsf{RP}, we have an efficient one-sided constant error algorithm $\mathcal{A}$ for $\textsc{ConsistentDecide}(\mathcal{C}, s)$. Using $O(c + \log \gamma) = \mathcal{O}(\poly(n))$ many calls to $\mathcal{A}$ and taking the majority, we can get error at most $\frac{1}{\gamma \cdot 2^c}$. Call this new algorithm $\mathcal{A}'$ and use it in place of the zero-error oracle for $\textsc{ConsistentDecide}(\mathcal{C}, s)$. By the union bound over all $\gamma$ calls to $\mathcal{A}'$, the probability that any query to $\mathcal{A'}$ differs from the zero-error oracle is at most $\frac{1}{2^c}$.
		
		By the union bound over both the samples and the error in $\mathcal{A}'$, the total error probability is at most $\delta + \frac{1}{2^c}$.
	\end{proof}

	\section{PAC Learning Applied to Clifford Circuits}\label{sec:pac_applied_clifford}
	Because of the works of \citet{rocchetto2018stabiliser} and \citet{Lai2022}, Clifford circuits are a prime candidate for an efficiently PAC-learnable class of circuits. We give a very loose bound on the fat-shattering dimension of Clifford circuits that is sufficient for our purposes.

	\begin{lemma}\label{lem:clifford_fat_dim_finite}
		Let $\mathcal{C}$ be the set of Clifford circuits. For all $\eta > 0$, $fat_\mathcal{C}(\eta) \leq \mathcal{O}(n^2)$.
	\end{lemma}
	\begin{proof}
		By \cref{prop:clifford_circuit_number} there are at most $2^{\mathcal{O}(n^2)}$ Clifford circuits. By \cref{lem:fat_dim_finite}, the fat-shattering dimension for all $\eta > 0$ is at most $\log_2\mleft(2^{\mathcal{O}(n^2)}\mright) = \mathcal{O}(n^2)$.
	\end{proof}

	Because CNOT circuits are a subset of Clifford circuits, we can also upper-bound the fat-shattering dimension of CNOT circuits by $\mathcal{O}(n^2)$.

	\subsection{Consistency Problem of Clifford Circuits}
	
	We now turn to the consistency problem. Noting that each Pauli matrix is Hermitian, a very natural way to measure a stabilizer state is in a product basis where we measure each qubit with respect to a Pauli.
	
	\begin{definition}
		If $P \in \mathcal{P}_n$ is a Pauli operator, then the two-outcome measurement associated with $P$ is $\frac{I^{\otimes n} + P}{2}$, and is referred to as a \textsc{Pauli measurement}.
	\end{definition}
	
	\begin{definition}\label{def:pac_clifford}
		Let the problem of PAC learning Clifford circuits with respect to Pauli measurements be defined as follows. Let $C$ be an unknown Clifford circuit and let $\mathcal{D}$ be an unknown joint distribution over both stabilizer states and Pauli measurements. Finally, let samples to $C$ be given as
		\[\mleft(\rho, E, \trc\mleft[E C \rho C^\dagger\mright]\mright)\]
		where $\rho, E \sim \mathcal{D}$ are a stabilizer state and Pauli measurement jointly drawn from $\mathcal{D}$ and represented as classical bit strings using the stabilizer formalism. The goal is to then learn the measurements $\trc\mleft[E C \rho C^\dagger\mright]$ up to error $\epsilon$ under the distribution $\mathcal{D}$.
	\end{definition}
	
	A critical part of \citet{rocchetto2018stabiliser} was noting that the measurement results with Pauli measurements could only have three distinct values:
	\begin{lemma}[\citet{rocchetto2018stabiliser} Lemma 1]\label{lemma:pauli_measure_trace}
		Let $E^P = \frac{I^{\otimes n} + P}{2}$ be a Pauli measurement associated to a Pauli operator $P \in \mathcal{P}_n$ and $\rho$ be an $n$-qubit stabiliser state. Then $\trc\mleft[E^P C\rho C^\dagger\mright]$ can only take on the values $\mleft\{0, \frac{1}{2}, 1\mright\}$, and:
		\begin{align*}
			\begin{cases}
				\trc\mleft[E^P C\rho C^\dagger\mright] = 1 \text{ iff $P$ is a stabilizer of $C \rho C^\dagger$;}\\
				\trc\mleft[E^P C\rho C^\dagger\mright] = 1/2 \text{ iff neither $P$ nor $-P$ is a stabilizer of $C \rho C^\dagger$;}\\
				\trc\mleft[E^P C\rho C^\dagger\mright] = 0 \text{ iff $-P$ is a stabilizer of $C \rho C^\dagger$.}
			\end{cases}
		\end{align*}
	\end{lemma}
	\noindent What information does a single sample tell us? Let $G_i$ be the stabilizer group of $\rho_i$. From this, we can gather that if $\trc\mleft[E^P C\rho_i C^\dagger\mright] = 1$ then $C^\dagger P C \in G_i$, and if $\trc\mleft[E^P C\rho_i C^\dagger\mright] = 0$ then $C^\dagger P C \in -G_i$ where $-G_i = \{-P \mid P \in G_i\}$. Finally, if $\trc\mleft[E^P C\rho_i C^\dagger\mright] = \frac{1}{2}$ then $C^\dagger P C$ is in the complement of $G_i \cup -G_i$.
	
	If the measurement $E^P$ appears multiple times across multiple samples, we can gather further information.
	For instance, have
	\[S_P = \mleft\{\mleft(\rho_i, E^P, \trc\mleft[E^P C\rho C^\dagger\mright]\mright)\mright\}\]
	be the set of all samples such that $E^P$ is the measurement taken and let $G_i$ be the stabilizer group of each $\rho_i$. Based on each label $\trc\mleft[E^P C\rho_i C^\dagger\mright]$, we know that $C^\dagger P C$ must lie in $H_i$, which is one of $G_i$, $-G_i$ or $G_i \cup -G_i$. We then deduce that $C^\dagger P C$ must lie in $\bigcap_i H_i$. To actually be a Clifford circuit, we must also add the constraint that $C^\dagger P C \neq I^{\otimes n}$, giving us
	\[C^\dagger P C \in \mleft(\bigcap_i H_i\mright)\setminus \mleft\{I^{\otimes n}\mright\}.\]
	The problem of finding a Clifford circuit with zero training error then reduces to the search problem of finding a set of $\alpha, \beta, \gamma, \theta, p, q$ from \cref{eq:clifford_gen} representing a $C^\dagger$ that is consistent with all of these constraints while remaining symplectic according to \cref{eq:symplectic_matrix}. Let $\mathcal{C}$ be the set of Clifford circuits. We will call this problem $\textsc{CliffordSearch}(s) = \textsc{ConsistentSearch}(\mathcal{C}, s)$\footnote{We stress that via \cref{def:pac_clifford}, the input states are assumed to be stabilizer states.}.

	Due to Gottesman-Knill \citep{https://doi.org/10.48550/arxiv.quant-ph/9807006, aaronson2004improved} showing that Clifford circuits are classically simulable, the act of verifying that we have a circuit that has zero training error is efficient, meaning that the decision version $\textsc{CliffordDecide}(s) = \textsc{ConsistentDecide}(\mathcal{C}, s)$ of the problem is in \np.
	\fussy
	
	\begin{proposition}\label{prop:in_np}
		\sloppy
		The decision problem, $\textsc{CliffordDecide}(\text{poly}(n))$, of deciding if there exists a Clifford circuit consistent with polynomially sized sample $S$ is in \np.
	\end{proposition}
	\begin{proof}
		Given a set of $\alpha, \beta, \gamma, \theta, p, q$, it easy to check that the $\gamma, \theta, p, q$ form a symplectic matrix by checking with respect to $\Lambda(n)$. Checking that they are consistent with the samples in $S$ can be done by iterating through $S$ since the trace can be computed efficiently using Gottesman-Knill \citep{aaronson2004improved}.
	\end{proof}
	
	\sloppy \noindent Knowing this, we find that $\textsc{CliffordSearch}(\text{poly}(n)) \in \mathsf{FNP}$. This property extends to the analogous problems for CNOT circuits, $\textsc{CNOTSearch}(\text{poly}(n))$ and $\textsc{CNOTDecide}(\text{poly}(n))$, since one can also efficiently verify that $\gamma_{ij} = 0$ and $p_j = 0$ for all $i$ and $j$.
	
	\fussy
	\section{Generating Samples with Certain Constraints} \label{sec:gen_constraint}
	We will now show how we can use samples from PAC learning to generate certain kinds of constraints. It will suffice to only consider CNOT circuits with computational basis state measurements and measurements of the form $\mleft\{I, Z\mright\}^{\otimes n}$. The net effect of this is that from a PAC learning standpoint, for unknown CNOT circuit $C$ we only need to figure a set of $C^\dagger Z_i C$ that is consistent with the samples as described in \cref{sec:pac_applied_clifford}. Since we will never be tested on a measurement with some component of $X_i$ involved, this is equivalent to finding the $\theta_{ij}$ and $q_j$ values from \cref{eq:clifford_gen} of $C^\dagger$. We will again choose to view the $\theta_{ij}$ as the matrix $\Theta$, such that $\Theta$ must be full rank.
	
	\begin{definition}
		Given a set of abelian generators $\{P_i\}$, let \[\rho(P_1, P_2, \dots, P_n) = \frac{1}{2^n}\sum_{P \in \langle P_1, \cdots, P_n \rangle} P\] be the stabilizer state that is formed from that stabilizer group.
	\end{definition}
	
	The following observation will notationally make the proceeding theorem states and proofs easier to follow.
	\begin{observation}
		Any one-dimensional affine subspace $v + \langle w \rangle$ can be represented as $\{v, v+w\}$ and any set of two vectors/matrices $\{v, w\}$ represents the one-dimensional affine subspace $v + \langle v+ w \rangle$. Thus we can freely move between the two representations.
	\end{observation}
	
	\begin{lemma}\label{lemma:create_sol_basic}
		Let $C$ be a CNOT circuit on $n$ qubits and have $\{v, v+w\} \subset \{0, 1\}^n$ be a one-dimensional affine subspace of column vectors such that $v \neq w$ and $w,v \neq 0$. Given an arbitrary pauli $P$ there exists a set of $n$ samples that constrains $C^\dagger P C$ to only have consistent solutions lying in $\{Z^v, Z^{v+w}\}$. Furthermore these $n$ samples can be efficiently generated.
	\end{lemma}
	\begin{proof}
		Let $(v, w, v_3, \cdots, v_n)$ be an arbitrary basis for $\{0, 1\}^n$ containing $v$ and $w$. This can be found with $\mathcal{O}(n)$ random vectors and the use of Gaussian elimination. Recalling \cref{def:power_z}, let us start by creating the sample
		\begin{align*}
			\mleft(\mleft(\rho\mleft(Z^{v}, Z^{w}, Z^{v_3}, Z^{v_4}, \dots, Z^{v_n}\mright), \frac{I^{\otimes n} + P}{2}\mright), 1 \mright),
		\end{align*}
		which limits $C^\dagger P C$ to be in $\{I, Z\}^{\otimes n}$ with positive phase.
		We can create the set of samples:
		\begin{align*}
			\bigg(\bigg(\rho(Z^{v}, Z^{w}, -Z^{v_3}, &Z^{v_4}, \dots, Z^{v_n}), \frac{I^{\otimes n} + P}{2}\bigg), 1 \bigg),\\
			\bigg(\bigg(\rho(Z^{v}, Z^{w}, Z^{v_3}, -&Z^{v_4}, \dots, Z^{v_n} ), \frac{I^{\otimes n} + P}{2}\bigg), 1 \bigg),\\
			&\vdots\\
			\bigg(\bigg(\rho(Z^{v}, Z^{w}, Z^{v_3}, &Z^{v_4}, \dots, -Z^{v_n}), \frac{I^{\otimes n} + P}{2}\bigg), 1\bigg).
		\end{align*}
		
		\noindent By construction $C^\dagger P C$ cannot have any component of $Z^{v_3}$ because of the first sample of this set, nor any $Z^{v_i}$ for $i > 3$ due to the remaining samples. This leaves $C^\dagger P C$ to be one of $Z^{v}$, $Z^{w}$, or $Z^{v + w}$ (since it cannot be identity). To remedy this, we can introduce the final sample:
		\begin{align*}
			\bigg(\bigg(\rho(-Z^{v}, Z^{w}, Z^{v_3}, Z^{v_4}, \dots Z^{v_n}, \frac{I^{\otimes n} + P}{2}\bigg), 0\bigg),
		\end{align*}
		which then eliminates $Z^{w}$ (and identity, due to the negative sign).
		The total number of samples is $n$ and the whole process takes polynomial in $n$ time to find the basis and create said samples.
	\end{proof}
	
	We can easily extend this to the $0$-dimensional case by simply treating $w$ as $v_2$, using an extra sample to remove the last dimension. More importantly, let's say we've constrained $C^\dagger Z^x C$ to lie in $\{Z^v, Z^{v+w}\}$. The effect of this on $\Theta$ is that if we sum the columns $i$ where $x_i = 1$ then the sum must lie in $\{v, v+w\}$.
	
	\begin{corollary}\label{cor:create_sol}
		Let 
		\[
		\{v, v+w\} =
		\mleft\{\begin{bmatrix}
			\vert & \vert & \vert & \vert\\
			v_1 & v_2& \dots& v_k\\
			\vert & \vert & \vert & \vert\\
		\end{bmatrix},
		\begin{bmatrix}
			\vert & \vert & \vert & \vert\\
			v_1 + w_1 & v_2 + w_2& \dots& v_k + w_k\\
			\vert & \vert & \vert & \vert\\
		\end{bmatrix}\mright\}
		\]
		be a one-dimensional affine subspace of $n \times k$ matrices over $\{0, 1\}$ such that for all $i$, $v_i \neq w_i$ and $v_i, w_i \neq 0$. Finally, let $\Theta'$ be an arbitrary $n \times k$ submatrix of $\Theta$. Then there exists a set of $(2k-1)n$ samples that constrain $\Theta'$ to only have consistent solutions lying in $\{v, v+w\}$ for CNOT circuit $C$. Furthermore these samples can be efficiently generated.
	\end{corollary}
	\begin{proof}
		WLOG, we will let the set of $k$ different columns we choose for $\Theta'$ to be columns $1$ through $k$. We will use induction on $k$ to prove this corollary, with the base case covered by \cref{lemma:create_sol_basic}. Now let us assume that we have samples that constrain columns $2$ through $k$ to be either
		\[
		\begin{bmatrix}
			\vert & \vert & \vert & \vert\\
			v_2& v_3& \dots& v_k\\
			\vert & \vert & \vert & \vert\\
		\end{bmatrix} \text{ or } \begin{bmatrix}\vert & \vert & \vert & \vert\\v_2 + w_2& v_2 + w_3& \dots & v_k + w_k\\ \vert & \vert & \vert & \vert\\\end{bmatrix}.
		\]
		The goal will be to generate constraints such that if column 2 is $v_2$ then column 1 must be $v_1$. Otherwise, if column 2 is $v_2 + w_2$ then column 1 is constrained to be $v_1 + w_1$.
		
		To start us off, we can use \cref{lemma:create_sol_basic} to constrain the sum of columns $1$ and $2$ to be either $v_1 + w_1$ or $v_1 + w_1 + v_2 + w_2$. If we focus on columns $1$ and $2$, the solutions to this constraint lie in an affine subspace defined by:
		\[
		\begin{bmatrix}
			\vert & \vert \\
			v_1+ v_2 + u& u\\
			\vert & \vert
		\end{bmatrix} \text{or}
		\begin{bmatrix}
			\vert & \vert \\
			v_1 + w_1 + v_2 + w_2 + u & u\\
			\vert & \vert
		\end{bmatrix}
		\]
		for arbitrary vector $u$.
		We then apply \cref{lemma:create_sol_basic} again to constrain column $1$ to be either $v_1$ or $v_1+w_1$.
		Thus the first two columns must either be
		\[
		\begin{bmatrix}
			\vert & \vert \\
			v_1& v_2\\
			\vert & \vert
		\end{bmatrix} \text{or}
		\begin{bmatrix}
			\vert & \vert \\
			v_1 + w_1  & v_2 + w_2\\
			\vert & \vert
		\end{bmatrix}
		\]
		Finally, to lie in the intersection from the inductive hypothesis, we note that if the second column is $v_2$ or $v_2 + w_2$ then columns $3$ through $k$ must be
		\[
		\begin{bmatrix}
			\vert & \vert & \vert & \vert\\
			v_3& v_4& \dots& v_k\\
			\vert & \vert & \vert & \vert\\
		\end{bmatrix} \text{or} \begin{bmatrix}
			\vert & \vert & \vert & \vert\\
			v_3+w_3 & v_4+w_4 & \dots& v_k+w_k\\
			\vert & \vert & \vert & \vert\\
		\end{bmatrix}.
		\]
		respectively.
		
		Collectively, we achieve our goal of constraining the entire solution to lie in $v + \langle w \rangle$. We used $n$ samples at the first step and $2n$ for every inductive step after (one set of $n$ samples for each call of \cref{lemma:create_sol_basic}), giving us a total number of samples of $2n(k-1) + n = 2kn - 2n + n = (2k-1)n$. Since each step was efficient, the whole process takes polynomial in $n$ time to generate all of the samples.
	\end{proof}
	
	\section{On the \np-completeness of \textsc{NonSingularity}}
	
	\begin{definition}
		Given $n \times n$ matrices $M_0, M_1, \dots, M_m$ over some field $\mathbb{F}$, \textsc{NonSingularity} is the problem of deciding if there exists $\alpha_1, \alpha_2, \dots, \alpha_m \in \mathbb{F}$ such that $M_0 + \sum_i \alpha_i M_i$ results in a non-singular matrix.
	\end{definition}
	
	\begin{theorem}[\citet{BUSS1999572} Corollary 10]
		\textsc{NonSingularity} over $\mathbb{F}_2$ is \np-complete.
	\end{theorem}
	
	The high level idea of the proof is to first reduce a 3SAT instance over variables $\{x_i\}$ to solving an arithmetic formula $F$. The formula is then turned into a weighted directed graph whose adjacency matrix $M(x)$ has a determinant that is equal to the formula $F$, where $M(x)$ has entries from $\{0, 1\} \cup \{x_i\}$, and can thus be viewed as an affine subspace over $\mathbb{F}_2^{(|F|+2) \times (|F| + 2)}$.
	
	While we will not prove the correctness of this statement, we will want to ascertain exactly what kind of $M_i$ are formed through the reduction. We now describe the construction of the graph (see \cref{fig:table} and \cref{fig:example} for relevant illustrations):
	
	\begin{figure}
		\centering
		\renewcommand{\arraystretch}{5}
		\begin{tabular}{|c|c|}
			\hline
			Formula $F$ & The series-parallel $s$-$t$ graph $G_F$ with edge withs\\
			\hline
			Constant $c$ & \begin{tikzpicture}[arrows={-Triangle[angle=30:10pt]},baseline=(a2.south)]
				\node [circle, draw, fill=black,scale=0.4,label=left:$s$] (a1) at (0,0) {};
				\node [circle, draw, fill=black,scale=0.4] (a2) at (2,0) {};      
				\node [circle, draw, fill=black,scale=0.4,label=right:$t$] (a3) at (4,0) {};
				\draw (a1) -- node [above] {$c$}(a2);
				\draw (a2) -- node [above] {$1$}(a3);
			\end{tikzpicture}\\
			Variable $x$ & \begin{tikzpicture}[arrows={-Triangle[angle=30:10pt]},baseline=(a2.south)]
				\node [circle, draw, fill=black,scale=0.4,label=left:$s$] (a1) at (0,0) {};
				\node [circle, draw, fill=black,scale=0.4] (a2) at (2,0) {};      
				\node [circle, draw, fill=black,scale=0.4,label=right:$t$] (a3) at (4,0) {};
				\draw (a1) -- node [above] {$x$}(a2);
				\draw (a2) -- node [above] {$1$}(a3);
			\end{tikzpicture}\\
			$F = F_1 \cdot F_2$ & \begin{tikzpicture}[arrows={-Triangle[angle=30:10pt]},baseline=(a2.south)]
				\node [ellipse, draw, fill=gray!20, anchor=west, minimum width=2cm] (a4) at (0, 0) {$G_{F_1}$};
				\node [ellipse, draw, fill=gray!20, anchor=west, minimum width=2cm] (a5) at (2, 0) {$G_{F_2}$};
				\node [circle, draw, fill=black,scale=0.4,label=left:$s\equal s_1$] (a1) at (0,0) {};
				\node [circle, draw, fill=black,scale=0.4, label={[label distance = 0.4cm]below:$t_1\equal s_2$}] (a2) at (2,0) {};      
				\node [circle, draw, fill=black,scale=0.4,label=right:$t\equal t_2$] (a3) at (4,0) {};
			\end{tikzpicture}\\
			$F = F_1 + F_2$ & \begin{tikzpicture}[arrows={-Triangle[angle=30:10pt]},baseline=(a1.south)]
				\node [ellipse, draw, fill=gray!20, anchor=west, minimum width=4cm, minimum height=2cm,label={[label distance = -0.8cm]above:$G_{F_1}$}] (a4) at (0, 0) {};
				\node [ellipse, draw, fill=white, anchor=west, minimum width=4cm, minimum height=0.5cm,label={below:$G_{F_2}$}] (a4) at (0, 0) {};
				\node [circle, draw, fill=black,scale=0.4,label=left:$s\equal s_1 \equal s_2$] (a1) at (0,0) {};    
				\node [circle, draw, fill=black,scale=0.4,label=right:$t\equal t_1 \equal t_2$] (a3) at (4,0) {};
				\node[circle] (a6) at (0, -1.5) {};
			\end{tikzpicture}\\
			\hline
		\end{tabular}
		\caption{Inductive Construction from Formula to Graph}
		\label{fig:table}
	\end{figure}
	
	\begin{figure}
		\centering
		\[F: x_1(x_2 + x_3) + x_3 \cdot x_4\]
		\begin{tikzpicture}[arrows={-Triangle[angle=30:10pt]}]
			\node [circle, draw, fill=black, scale=0.4, label=left:$s$] (a1) at (0, 0) {};
			\node [circle, draw, fill=black, scale=0.4] (a2) at (2, 0) {};
			\node [circle, draw, fill=black, scale=0.4] (a3) at (4, 0) {};
			\node [circle, draw, fill=black, scale=0.4] (a4) at (6, 1) {};
			\node [circle, draw, fill=black, scale=0.4] (a5) at (6, -1) {};
			\node [circle, draw, fill=black, scale=0.4] (a6) at (2.5, 2) {};
			\node [circle, draw, fill=black, scale=0.4] (a7) at (5, 3) {};
			\node [circle, draw, fill=black, scale=0.4] (a8) at (7.25, 2) {};
			\node [circle, draw, fill=black, scale=0.4, label=right:$t$] (a9) at (8, 0) {};
			
			\draw (a1) -- node [above] {$x_1$}(a2);
			\draw (a2) -- node [above] {$1$}(a3);
			\draw (a3) -- node [above] {$x_2$}(a4);
			\draw (a3) -- node [above] {$x_3$}(a5);
			\draw (a4) -- node [above] {$1$}(a9);
			\draw (a5) -- node [above] {$1$}(a9);
			\draw (a1) -- node [above] {$x_3$}(a6);
			\draw (a6) -- node [above] {$1$}(a7);
			\draw (a7) -- node [above] {$x_4$}(a8);
			\draw (a8) -- node [above] {$1$}(a9);
			\draw (a9) .. controls (6,-3) .. node [below] {$1$}(a1);
			\path (a2) edge [loop above, min distance=1cm, in=60, out=120] node {$1$} (a2);
			\path (a3) edge [loop above, min distance=1cm, in=60, out=120] node {$1$} (a3);
			\path (a4) edge [loop above, min distance=1cm, in=60, out=120] node {$1$} (a4);
			\path (a5) edge [loop above, min distance=1cm, in=60, out=120] node {$1$} (a5);
			\path (a6) edge [loop above, min distance=1cm, in=60, out=120] node {$1$} (a6);
			\path (a7) edge [loop above, min distance=1cm, in=60, out=120] node {$1$} (a7);
			\path (a8) edge [loop above, min distance=1cm, in=60, out=120] node {$1$} (a8);
			\path (a9) edge [loop above, min distance=1cm, in=60, out=120] node {$1$} (a9);
		\end{tikzpicture}
		\[
		\begin{pmatrix}
			0 & x_1 & 0 & 0 & 0 & x_3 & 0 & 0 & 0\\
			0 & 1 & 1 & 0 & 0 & 0 & 0 & 0 & 0\\
			0 & 0 & 1 & x_2 & x_3 & 0 & 0 & 0 & 0\\
			0 & 0 & 0 & 1 & 0 & 0 & 0 & 0 & 1\\
			0 & 0 & 0 & 0 & 1 & 0 & 0 & 0 & 1\\
			0 & 0 & 0 & 0 & 0 & 1 & 1 & 0 & 0\\
			0 & 0 & 0 & 0 & 0 & 0 & 1 & x_4 & 0\\
			0 & 0 & 0 & 0 & 0 & 0 & 0 & 1 & 1\\
			1 & 0 & 0 & 0 & 0 & 0 & 0 & 0 & 1\\
		\end{pmatrix}
		\]
		
		\caption{Example of constructing the adjacency matrix with a specific determinant}
		\label{fig:example}
	\end{figure}
	
	\begin{itemize}
		\item For each atomic formula $F'$, create vertices $s$ and $t$.
		\subitem For each constant $c$ create a unique node $v_c$ with a path from $s$ to $v_c$ with weight $c$ and a path from $v_c$ to $t$ with weight $1$.
		\subitem For each variable $x_i$ create a unique node $v_{x_i}$ with a path from $s$ to $v_{x_i}$ with weight $x_i$ and a path from $v_{x_i}$ to $t$ with weight $1$.
		\item For multiplication of $F_i$ and $F_j$, place the graphs of $F_i$ and $F_j$ in series.
		\item For addition of $F_i$ and $F_j$, place the graphs of $F_i$ and $F_j$ in parallel.
		\item Once all of this is done, create a path of weight $1$ from the global $t$ vertex to the global $s$ vertex.
		\item Create self loops at every vertex besides the global $s$ vertex.
	\end{itemize}
	
	Let $M$ be the resulting adjacency matrix of this graph. For every entry that is a constant, we can assign that to $M_0$. Then for each variable $x_i$, we can set $M_i$ to be the matrix that is zero everywhere except where $x_i$ appears in $M$. As an example, for a matrix 
	\begin{align*}
		M(x) = \begin{pmatrix}x_1 & 1\\ x_2 & x_1 \end{pmatrix}
	\end{align*}
	
	we can describe it using $M(x) = M_0 + x_1 M_1 + x_2 M_2$ where
	\begin{align*}
		M_0 = \begin{pmatrix}0 & 1\\ 0 & 0 \end{pmatrix}\,\,\,\,\,
		M_1 = \begin{pmatrix}1 & 0\\ 0 & 1 \end{pmatrix}\,\,\,\,\,
		M_2 = \begin{pmatrix}0 & 0\\ 1 & 0 \end{pmatrix}
	\end{align*}

	For a more succinct reduction later in \cref{sec:pac_np_hard}, we want to isolate the kinds of matrix affine subspaces over $\mathbb{F}_2$ that are hard to solve (i.e., are used in the reduction from 3SAT). The following notation will be beneficial for that.

	\begin{definition}
		For a $n \times n$ matrix $M$ over a field $\mathbb{F}$, let $NZ(M) \subseteq \{1, 2, \cdots, n\}$ be the columns of $M$ that are non-zero (i.e., are not the all zeros vector).
	\end{definition}

	\begin{definition}
		Let $M$ and $W$ be $n \times n$ matrices over a field $\mathbb{F}$ and let $k \coloneqq |NZ(W)|$ such that $NZ(W) = \{i_1, i_2, \cdots i_k\}$. If 
		\[M = 
		\begin{bmatrix}
			\vert & \vert & \vert & \vert\\
			m_{1}& m_{2 }& \dots& m_{n}\\
			\vert & \vert & \vert & \vert\\
		\end{bmatrix}
		\]
		then the restriction of the columns of $M$ to the non-zero columns of $W$ is defined as
		\[
		R_W(M) = \begin{bmatrix}
			\vert & \vert & \vert & \vert\\
			m_{i_1}& m_{i_2 }& \dots& m_{i_{k}}\\
			\vert & \vert & \vert & \vert\\
		\end{bmatrix}.
		\]
		Likewise, for a set of matrices $S$, $R_W(S) = \{R_W(M) \mid M \in S\}$.
	\end{definition}

	We now define the modified version of \textsc{NonSingularity} and show that it is still \np-complete.

	\begin{definition}\label{def:mod_sing}
		Let $M = M_0 + \langle M_1, M_2, \cdots M_m \rangle$ be an matrix affine subspace of $n \times n$ matrices over some field $\mathbb{F}$. In addition, require the $M_i$ have the property that $NZ(M_i) \cap NZ(M_j) = \emptyset$ for all $i > 0$ and $j > 0$. Finally, for all $i > 0$ with $k_i \coloneqq |NZ(M_i)|$, constrain $M_0$ and $M_i$ such that the restriction of $M$ to the non-zero columns of $M_i$, \[R_{M_i}(M) = R_{M_i}(\{M_0, M_0 + M_i\}),\] can be represented as:
		\[
		\mleft\{
		\begin{bmatrix}
			\vert & \vert & \vert & \vert\\
			v_1 & v_2& \dots& v_{k_i}\\
			\vert & \vert & \vert & \vert\\
		\end{bmatrix}, 
		\begin{bmatrix}
			\vert & \vert & \vert & \vert\\
			v_1 + w_1& v_2 + w_2& \dots& v_{k_i} + w_{k_i}\\
			\vert & \vert & \vert & \vert\\
		\end{bmatrix}\mright\}
		\]
		for some $v$ and $w$ such that $v_j \neq w_j$ and $v_j, w_j \neq 0$. The problem of deciding if there exists $\alpha_1, \alpha_2, \cdots, \alpha_m \in \mathbb{F}$ such that $M_0 + \sum_i \alpha_i M_i$ results in a non-singular matrix will be called the \textsc{Modified-NonSingularity} problem.
	\end{definition}
	
	\begin{corollary} \label{cor:disjoint_column}
		The \textsc{Modified-NonSingularity} problem is \np-complete over $\mathbb{F}_2$.
	\end{corollary}
	\begin{proof}
		\sloppy
		Rather than solve \textsc{NonSingularity} using \textsc{Modified-NonSingularity}, we merely show that the reduction from 3SAT naturally leads to instances of \textsc{Modified-NonSingularity}.
		
		Let $G$ be the graph produced by the reduction from 3SAT with adjacency matrix $M(x)$.
		By the construction of $M(x)$ given by \cref{fig:table}, we see that every instance of a variable will create its own unique subgraph such that the instance of each variable connects to a unique vertex. Because this unique vertex can never be used as an $s$ or $t$ vertex from \cref{fig:table}, that vertex also necessarily has in-degree $1$ so that no other edges connect to it. Because these are the only vertices that have incoming edges assigned with variable weight, this conversely means that the columns of $M(x)$ contain at most one variable.
		
		Recall that for variable $x_i$ with $i > 0$, the matrix $M_i$ we form from the decomposition of $M(x)$ is the entries associated with $x_i$. Since each column only contains at most one $x_i$, then $NZ(M_i) \cap NZ(M_j)$ must be disjoint for $i > 0$ and $j > 0$. Due to the constraint being met, we know that the $R_{M_i}(M)$ is confined to a one-dimensional affine subspace.
		As such, it can be represented as:
		\[
		\mleft\{
		\begin{bmatrix}
			\vert & \vert & \vert & \vert\\
			v_1 & v_2& \dots& v_{k_i}\\
			\vert & \vert & \vert & \vert\\
		\end{bmatrix}, 
		\begin{bmatrix}
			\vert & \vert & \vert & \vert\\
			v_1 + w_1& v_2 + w_2& \dots& v_{k_i} + w_{k_i}\\
			\vert & \vert & \vert & \vert\\
		\end{bmatrix}\mright\}
		\]
		for some $v$ and $w$, where the $v$ is a sub-matrix of $M_0$ and $w$ is the corresponding sub-matrix for $M_i$.
		Due to the definition of $NZ(M_i)$ it is clear that the $w_j$ are non-zero.
		Furthermore, since the weight of an edge will never be $x_i + 1$ then an entry of $w_j$ being one implies the corresponding entry of $v_j$ is zero. This implies $v_j \neq w_j$ for all $j$.
		Finally, to ensure that $v_j \neq 0$, we note that each vertex besides $s$ receives a self loop with weight $1$. $s$ instead receives an edge from $t$ with weight $1$. These self loops and the edge from $t$ to $s$ ensures that each column of $M_0$ has at least entry with $1$ in it such that the $v_j \neq 0$.
		
		We have now shown that every matrix affine subspace produced from the reduction in \citet{BUSS1999572} also meets the requirements for \textsc{Modified-NonSingularity}, thus showing that this problem is also \np-complete.
	\end{proof}
	
	\section{PAC Learning CNOT Circuits and \np}\label{sec:pac_np_hard}
	
	\begin{lemma}\label{lemma:cnot_np_hard}
		The decision problem, $\textsc{CNOTDecide}(2n^2)$, of deciding whether or not there exists a CNOT circuit consistent with at most $2n^2$ samples is \np-complete.
	\end{lemma}
	\begin{proof}
		By \cref{cor:disjoint_column} and the Cook-Levin theorem \citep{kleinberg2022algorithm}, if a problem can be used to solve \textsc{Modified-NonSingularity} then it is \np-hard. Looking at each individual $M_i$ for $i > 0$ from \cref{def:mod_sing}, they are all non-zero on disjoint columns. If we restrict to the non-zero columns of a particular $M_i$ for $i > 0$, we find that $R_{M_i}(M)$ meet the requirements for \cref{cor:create_sol}. Since the $M_i$ act on disjoint columns, if we apply \cref{cor:create_sol} for each $R_{M_i}(M_i)$ then we have efficiently created a set of samples $S_i$ that restricts those columns of $\Theta$, the matrix form of the $\theta_{ij}$ values, to lie in $R_{M_i}\mleft(\{M_0, M_0 + M_i \}\mright)$. To fix the columns not touched by the $M_i$, it is not hard to show that the ideas of \cref{lemma:create_sol_basic} can also create a 0-dimensional affine space over these columns, thereby setting the whole matrix to lie in $M_0 + \langle M_1, M_2, \cdots M_m \rangle$. Let such samples be called $T$. Altogether, by taking $S = T \cup \mleft(\bigcup_i S_i\mright)$ we are able to use \cref{lemma:create_sol_basic} and \cref{cor:create_sol} to set $\Theta$ to lie in $M(x)$ from \textsc{Modified-NonSingularity}.
		
		If $M(x)$ is an accepting instance of \textsc{Modified-Singularity} then there must be a full rank $\Theta$ that is consistent with $S$. Since \cref{lemma:theta_rowsum} ensures us that CNOT circuits can instantiate any full rank $\Theta$, there must also exist a CNOT circuit consistent with $S$. Alternatively, if $M(x)$ does not contain a non-singular matrix, then there does not exist a full rank theta that is consistent with $S$. By \cref{lemma:theta_full_rank} there cannot exist a CNOT circuit that is consistent with the data. This gives us that $M(x)$ contains a non-singular matrix if and only if there is a CNOT circuit consistent $S$, which can be produced efficiently.
		
		We now count the number of samples used. For each all $1 \leq i \leq m$, we use $n$ samples to constrain the first column of $NZ(M_i)$. For the remaining columns, we either use $2n$ if that column is contained in some $NZ(M_j)$, otherwise, we use $n+1$ samples from the generalization of \cref{lemma:create_sol_basic}. Since $2n >= n+1$ for $n \geq 1$, we use at most $2n (n-m)$ samples for the remainder, giving us at most $2n(n-m) + nm = 2n^2 - mn \leq 2n^2$ samples.
		This shows $\textsc{CNOTDecide}(2n^2)$ is \np-hard. Combined with \cref{prop:in_np} we find that $\textsc{CNOTDecide}(2n^2)$ to be \np-complete.
	\end{proof}
	
	Since $\textsc{CNOTDecide}(2n^2) \subset \textsc{CNOTDecide}(\poly(n))$ then  $\textsc{CNOTDecide}(\poly(n))$ is also \np-hard.
	
	\begin{theorem}[Formal Statement of \cref{thm:main}]\label{thm:main_formal}
		There exists an efficient randomized $(\epsilon, \delta)$ proper PAC learner for CNOT circuits with arbitrary $\epsilon$ and $\delta$ as arbitrary $\frac{1}{\poly(n)}$ values if and only if $\textsf{RP} = \np$. Furthermore, an efficient quantum $(\epsilon, \delta)$ proper PAC learner for CNOT circuits with arbitrary $\epsilon$ and $\delta$ as arbitrary $\frac{1}{\poly(n)}$ values exists if and only if $\np \subseteq \textsf{RQP}$.
	\end{theorem}
	\begin{proof}
		We will start by proving the $\np \subseteq \textsf{RP}$ version for classical randomized learners. The quantum version will follow trivially by replacing the learner with a quantum algorithm and therefore $\textsf{RP}$ with $\textsf{RQP}$. The only change then is that $\np \subseteq \textsf{RQP}$ does not necessary imply $\np = \textsf{RQP}$ like it does with \textsf{RP}.
		
		\sloppy Because CNOT circuits with classical inputs and measurements only has labels $0$ and $1$, the smallest non-zero error is $1$. By \cref{prop:decision2learning} with $\alpha = 1$, an efficient
		\[\mleft(\epsilon < \frac{1}{4n^4}, \delta < \frac{1}{2} + \frac{1}{4n^2}\mright)\footnote{We again abuse notation to signify that $\epsilon$ is a value less than $\frac{1}{4n^4}$ and likewise for $\delta < \frac{1}{2} + \frac{1}{4n^2}$.}\]
		randomized proper learner for CNOT circuits will imply $\textsc{CNOTDecide}(2n^2) \in \textsf{RP}$. Since $\textsc{CNOTDecide}(2n^2)$ is \np-complete by \cref{lemma:cnot_np_hard}, efficient randomized learners for CNOT circuits with arbitrary $\epsilon = \frac{1}{\poly(n)}$ and $\delta = \frac{1}{\poly(n)}$ are only possible if $\np \subseteq \textsf{RP}$.
		
		Conversely, by \cref{lem:clifford_fat_dim_finite} and \cref{lem:learning2decision} with $\beta = \frac{1}{2}$, $c = \poly(n)$, $\epsilon = \frac{4\epsilon'-1}{3}$, and $\delta = \delta' - 2^{-c}$ if $\np \subseteq \textsf{RP}$ then there must exist an efficient $(\epsilon', \delta')$ proper learner for CNOT circuits as long as the number of samples $m$ is polynomial in $n$. For arbitrary $\epsilon' = \frac{1}{\poly(n)}$ and $\delta' = \frac{1}{\poly(n)}$ our required number of samples becomes
		\[
		m = \Theta\mleft(\frac{1}{\epsilon}\mleft(fat_\mathcal{C}\mleft(\frac{\beta}{8}\mright)\log^2\mleft(\frac{fat_\mathcal{C}\mleft(\frac{\beta}{8}\mright)}{\beta\epsilon}\mright) + \log \frac{1}{\delta}\mright)\mright) = \mathcal{O}(\poly(n))
		\]
		This is sufficiently small as desired, completing the proof.
	\end{proof}
	
	\begin{corollary}[Formal Statement of \cref{cor:main}]
		There exists an efficient randomized $(\epsilon, \delta)$ proper PAC learner for Clifford circuits with arbitrary $\epsilon$ and $\delta$ as arbitrary $\frac{1}{\poly(n)}$ values if and only if $\textsf{RP} = \np$. Furthermore, an efficient quantum $(\epsilon, \delta)$ proper PAC learner for Clifford circuits with arbitrary $\epsilon$ and $\delta$ as arbitrary $\frac{1}{\poly(n)}$ values exists if and only if $\np \subseteq \textsf{RQP}$.
	\end{corollary}
	\begin{proof}
		\sloppy Since CNOT circuits are also a form of Clifford circuits, $\textsc{CNOTDecide}(2n^2) \subset \textsc{CliffordDecide}(2n^2)$, and so $\textsc{CliffordDecide}(2n^2)$ is also \np-hard. Combined with \cref{prop:in_np}, we get that it is \np-complete. The proof of \cref{thm:main_formal} continues but with $\alpha = \frac{1}{2}$ instead due to \cref{lemma:pauli_measure_trace}. This leads to slightly different constants, but the proof ideas all follow without major change.
	\end{proof}
	
	\section{Discussion and Open Problems}
	In this work, we prove a negative result in proper learning of one of the best candidates for efficient PAC learning of quantum circuits. However, it should be noted that in many cases there exist improper learners even in the case where proper learning is \np-hard, such as 2-clause CNF, 3-DNF, and intersection of half spaces\citep{blum2015dp, BLUM1992117, Haghtalab2020}. This immediately leaves the problem of whether or not an improper learner exists for Clifford circuits. One way of showing hardness would be to leverage cryptographic hardness such as in \citet{kharitonov_1993, arunachalam_2021}. Another approach would be to assume the hardness of random k-DNF, such as the work of \citet{daniely16}. For upper bounds, the work of \citet{2020Caro} can also be used to get agnostic generalization results, providing a possible pathway to answering research questions in that direction.
	
	Another thing worth considering is that the hardness results only apply for small errors (roughly $1/\poly(n)$). And while this is sufficient to give complexity-theoretic hardness for the kinds of PAC learners (i.e., strong proper learners) originally introduced by \citet{valiant1984theory}, it would be nice to get hardness results for larger errors as in the work by \citet{venkatesan_2009}. This work involved using PCP/hardness-of-approximation ideas to show that even constant training error was \np-hard.
	
	We also note that a single output bit of a CNOT circuit is simply an XOR, which is easy to learn efficiently if there is no noise. However, because we are dealing with reversible computation each output bit has to be a linearly independent XOR such that each input bit is recoverable. Finding a CNOT circuit that matches a single XOR function $f$ can be done by sampling expected $\mathcal{O}(n)$ random XOR until we get $n$ linearly independent XOR with one of them being $f$ (see \cref{ssec:single_measurement}). Thus, the entire difficulty of proper learning CNOT circuits is this linear independence of the output bits. As such, even though $\textsf{AC}^0 \subseteq \textsf{TC}^0 \subseteq \textsf{NC}^1 \subseteq \textsf{L} \subseteq \oplus \textsf{L}$ with the lower classes having improper hardness results based on cryptographic hardness \citep{kharitonov_1993}, one cannot directly give an improper learning result for CNOT circuits despite the fact that simulating CNOT circuits is complete for $\oplus \textsf{L}$.
	
	As noted previously, our PAC learning framework is slightly different from that of \citet{2020Caro}, in that we use Pauli matrices, rather than rank $1$ projectors as measurements. To the author's knowledge, there exists no proof showing that one framework is necessarily harder than the other. The author also do not see an obvious way of proving an analogous hardness theorem in the specific framework of \citet{2020Caro} for Clifford or CNOT circuits.
	
	Finally, with everything from the input states to the circuits involved being classical, it is entirely possible to prove the technical results about CNOT circuits only talking about bit strings and parity functions. Namely, one can replace the entire problem with samples of the form $(x, s, s^T Cx \mod 2)$ where $(x, s) \sim \mathcal{D}$ are $n$-bit strings and $C$ is a CNOT circuit. Since the stabilizer group of a computational basis state always lies in $\{\pm 1\}\times \{I, Z\}^{\otimes n}$, we can uniquely define it by the subgroup that has positive phase. This is equivalent to the orthogonal complement of $x$, which is the subspace $W_x = \mleft\{x \in \{0, 1\}^n \bigm| w \cdot x = 0 \mright\}$. From there, a sample of the form $(x, s, 0)$ simply says that $Cx \in W_s$, and one can get an analogous proof by copying the lemmas and theorems in \cref{sec:gen_constraint} and \cref{sec:pac_np_hard}. However, this proof isn't anymore intuitive than the one given using stabilizer groups, and in fact is probably less intuitive to the average reader due to the lack of established formalism from stabilizers and paulis. It would be interesting if a more intuitive purely classical proof could be made to show hardness of learning CNOT circuits under this model.
	
	\section*{Acknowledgements}
	
	The author would like to thank Aravind Gollakota for being a wealth of information with regard to statistical learning theory and William Kretschmer for insightful commentary and feedback. The author would also like to thank Patrick Rall, Justin Yirka, Sabee Grewal, Matthias Caro, and Scott Aaronson for useful discussions. The author was supported by Scott Aaronson's Paths to Quantum Supremacy award by the Office of Naval Research and Scott Aaronson's Simons Foundation It From Qubit award.
	
	\bibliography{../refs} 
	\bibliographystyle{plainnat}
	
	\appendix
	\newpage	
	
		\section{Generalization of \texorpdfstring{\citet{2020Caro}}{[Caro and Datta 2020]}}\label{sec:pseudo}
	We generalize the main results of \citet{2020Caro}, which is itself analogous to \citet{Goldberg1995}, to allow for projective measurements beyond rank 1, such as in the settings used in this paper. While not necessary for the main result of this paper, we hope to give this proof in a way that is also more black-box accessible for future work.
	
	\subsection{Quantum Circuits as Polynomials}
	
	The end goal will be to show that the outputs of our concept class can be described as a set of polynomials with bounded degree. Combined with an upper bound on the number of polynomials in that set, we can later arrive at an upper bound on the pseudo-dimension, which itself is an upper-bound on fat-shattering dimension.
	
	We now show a more terse version of Lemma 3 from \citet{2020Caro}.
	
	\begin{lemma}\label{lem:fixed_circuit_poly}
		Consider a quantum circuit $C$ with a fixed circuit structure (i.e., the location of the $2$-qudit gate are in fixed positions, though the gates themselves can be arbitrary) comprised of at most $\Gamma$ $2$-qudit gates. Such circuit can be described using variables $c_1, c_2, \cdots, c_k \in \mathbb{R}$ such that $k = 2\gamma d^4$. Then for every pair of quantum states $\ket{\psi}$ and $\ket{\phi}$ there exists a polynomial $p_{\mleft(\ket{\psi}, \ket{\phi}\mright)}(c_1, c_2, \cdots, c_k) = \trc\mleft[ C\ketbra{\psi}{\psi}C^\dagger \ketbra{\phi}{\phi} \mright]$ with degree at most $2 \Gamma$.
	\end{lemma}
	\begin{proof}
		Every $2$-qudit gate $U$ can be na\"ively expressed as the $d^4$ complex values that make up the $d^2 \times d^2$ unitary. By splitting up the complex values into a real and imaginary part, we get $2d^4$ real values to describe each $2$-qudit gate. If $\ket{\psi} = \sum_{s \in \{0, 1\}^n} \alpha_s \ket{s}$ then applying a $2$-qudit unitary to $\ket{\psi}$ leaves us with the amplitudes of this new state being a polynomial of degree $1$ in terms of the entries of $U$. Note that the $\alpha_i$, along with the circuit structure, are what determine the coefficients of this polynomial. By applying all $\gamma$ $2$-qudit gates that comprise $C$, the amplitudes of $C \ket{\psi}$ can be described as a polynomial of degree $\Gamma$ in $2\gamma d^4$ variables. Finally, since we can write $\ket{\phi} =  \sum_{s \in \{0, 1\}^n} \beta_s \ket{s}$, then the inner product $\braket{\phi| C |\psi}$ is some weighted linear combination of the amplitudes of $C \ket{\psi}$, which is a again polynomial with degree at most $\Gamma $. To get
		\[
		\trc\mleft[ C\ketbra{\psi}{\psi}C^\dagger \ketbra{\phi}{\phi} \mright] = \abs{\braket{\phi | C | \psi}}^2
		\]
		we note that the degree at most doubles when we multiply a polynomial by itself. This leaves us with $p_{\mleft(\ket{\psi}, \ket{\phi}\mright)}(c_1, c_2, \cdots, c_k)$ as polynomial of degree at most $2 \Gamma $ and $m = 2 \gamma d^4$.
	\end{proof}
	
	\begin{corollary}\label{cor:fixed_circuit_poly}
		Consider a quantum circuit $C$ with a fixed circuit structure (i.e., the location of the $2$-qudit gate are in fixed positions, though the gates themselves can be arbitrary) comprised of at most $\gamma$ $2$-qudit gates. Such circuit can be described using variables $c_1, c_2, \cdots, c_k \in \mathbb{R}$ such that $k = 2\gamma d^4$. Then for every pair of quantum state $\ket{\psi}$ and projector $\Pi$ there exists a polynomial $p_{\mleft(\ket{\psi}, \Pi\mright)}(c_1, c_2, \cdots, c_k) \coloneqq \trc\mleft[ C\ketbra{\psi}{\psi}C^\dagger \Pi \mright]$ with degree at most $2 \gamma$.
	\end{corollary}
	\begin{proof}
		We note that $\Pi = \sum_i \ketbra{\phi_i}{\phi_i}$. By linearity, of the trace
		\begin{align*}
			p_{\mleft(\ket{\psi}, \Pi\mright)}(c_1, c_2, \cdots, c_k)
			&= \trc\mleft[ C\ketbra{\psi}{\psi}C^\dagger \Pi \mright]\\
			&= \trc\mleft[ C\ketbra{\psi}{\psi}C^\dagger \sum_i \ketbra{\phi_i}{\phi_i} \mright]\\
			&= \sum_i \trc\mleft[ C\ketbra{\psi}{\psi}C^\dagger \ketbra{\phi_i}{\phi_i} \mright]\\
			&= \sum_i p_{\mleft(\ket{\psi}, \ket{\phi_i}\mright)}(c_1, c_2, \cdots, c_km).
		\end{align*}
		By \cref{lem:fixed_circuit_poly}, this is the sum of real polynomials in $2\Gamma d^4$ variables with degree at most $2 \gamma$. Since the sum does not increase the degree, we are done. 
	\end{proof}
	
	Since we fixed the circuit structure, we will want to know how many circuit structures there are, because this directly bounds the number of polynomials we need to consider. The following result was the main ingredient in the proof of Lemma 2 from \citet{2020Caro}.
	
	\begin{lemma}[\citet{2020Caro} Lemma 2]\label{lemma:circuit_structures_num}
		There are at most $\frac{\gamma! \delta^{\gamma - \delta}}{(\gamma - \delta)!}(n!)^\delta$ ways to structure $2$-qudit circuits with size $\gamma$ and depth $\delta$.
	\end{lemma}
	
	\subsection{Pseudo-Dimension of Concept Classes Described via Polynomials}
	
	The following is a generalization of \citet{Goldberg1995}, which used the degree of polynomials to bound the pseudo-dimension of concept classes that could be defined using polynomials in the parameter space of the concepts.
	
	\begin{definition}\label{def:pseudodim}
		The \textit{pseudo-dimension} of a concept class $\mathcal{C}$ is the limit of the fat-shattering dimension parameter $\eta$ as $\eta$ goes to zero. Formally, the pseudo-dimension is $\lim_{\eta \rightarrow 0^+} fat_\mathcal{C}(\eta)$.
	\end{definition}
	
	Because fat-shattering dimension increases as $\eta$ decreases, fat-shattering dimension is always upper-bounded by the pseudo-dimension for all values of $\eta > 0$.
	
	\begin{definition}
		Let $\{p_1, p_2, \cdots, p_m\} \subseteq \mathbb{R}^k \rightarrow \mathbb{R}$ be a set of $m$ polynomials on $k$ variables. For $\eta > 0$, the $\eta$-\textit{sign assignment} of $\{p_1, p_2, \cdots, p_m\}$ on the input $\mleft( x_1, x_2, \cdots, x_k \mright) \in \mathbb{R}^k$ is the vector $b \in \{-1, 0, 1\}^m$ such that
		\[
		b_i = \begin{cases}
			1 & p_i(x_1, x_2, \cdots, x_k) \geq \eta\\
			-1 & p_i(x_1, x_2, \cdots, x_k) \leq -\eta\\
			0 & \text{otherwise}
		\end{cases}.
		\]
	\end{definition}
	
	\begin{lemma}[\citet{Goldberg1995} Corollary 2.1]\label{lem:consistent_sign}
		Let $\{p_1, p_2, \cdots, p_m\} \subseteq \mathbb{R}^k \rightarrow \mathbb{R}$ be a set of real polynomials in $k$ variables with $m \geq k$, each of degree at most $d \geq 1$. Then the number of unique $\eta$-sign assignments that $\{p_1, p_2, \cdots, p_m\}$ can create over all inputs in $\mathbb{R}^k$ is at most $\mleft(\frac{8edm}{k}\mright)^k$ in the limit as $\eta \rightarrow 0^+$.
	\end{lemma}

	\cref{lem:consistent_sign} is a stronger notion than pseudo-dimension, since it upper bounds the number of sign assignments over arbitrarily large sets of inputs.
	Since pseudo-dimension requires that $\mathcal{C}$ can achieve all (i.e., an exponential number of) sign assignments on some large set of samples, we can show that the pseudo-dimension cannot be too large.
	We formalize that notion here.
	Note that the polynomials in question in the following proof are over the parameters of the concept class, not the inputs.
	The intuition is that if the output of the concept is some bounded-degree polynomial in the parameter space, there cannot be too many sign assignments.
	
	\begin{theorem}[Generalization of \citet{Goldberg1995} Theorem 2.2]\label{thm:psuedo_dim_poly}
		Let $\mathcal{C} \subseteq \Omega \rightarrow [0, 1]$ be a concept class such that every element of $\mathcal{C}$ can be described via $k$ different real variables $c_1, c_2, \cdots c_k \in \mathbb{R}$, as well as an index $l \in [s]$ for $s \geq 0$. Furthermore, for every $f_{c_1, c_2, \cdots c_k, l} \in \mathcal{C}$ and $x \in \Omega$, let $f_{c_1, c_2, \cdots c_k, l}(x) = p_{x, l}(c_1, c_2, \cdots c_k)$ where $p_{x, l}$ is one of $s$ polynomials each with degree at most $d$ for $d \geq 1$. Then the pseudo-dimension of $\mathcal{C}$ is at most $2k \log_2 \mleft(8 eds\mright)$.
	\end{theorem}
	\begin{proof}
		Let $(x_1, y_1), (x_2, y_2), \cdots, (x_m, y_m)\ \subseteq \mathbb{R}$ be the largest set of points pseudo-shattered by $\mathcal{C}$.
		If $ms < k$, then there is no issue because the largest shattered set is smaller than $k$, which is smaller than $2k \log_2 \mleft(8 eds\mright)$.
		Now assume that $ms \geq k$.
		By \cref{def:pseudodim}, there must exist some points $y_1, y_2, \cdots, y_m \in \mathbb{R}$ and some (potentially arbitrarily small) value $\eta > 0$ such that for all $b_i \in \{\pm 1\}^m$, there is a $f_{c_1, c_2, \cdots c_k} \in \mathcal{C}$ with $b_i \cdot \mleft(f_{c_1, c_2, \cdots c_k}(x_i) - y_i \mright) \geq \eta$.
		Because $f_{c_1, c_2, \cdots c_k, l}(x_i) = p_{x_i, l}(c_1, c_2, \cdots c_k)$, we can define the new set of polynomials $p'_{i, l} = p_{x_i, l} - y_i$.
		This means that $\bigcup_{i, l} \{p'_{i, l}\}$  is a set of $ms$ polynomials that must be able to create at all $2^m$ different sign assignments that define $b$.
		However, we know from \cref{lem:consistent_sign} that the number of different sign assignments is at most $\mleft(\frac{8edms}{k}\mright)^k$ as long as $ms \geq k$, which we have assumed to be true.
		Therefore, $2^m \leq \mleft(\frac{8edms}{k}\mright)^k$.
		Taking the logarithm of both sides, \[m \leq k \log_2 \mleft(\frac{8edms}{k}\mright) = k \log_2 \mleft(8eds\mright) + k\log_2 \frac{m}{k}.\]
		We divide the situation into two cases based on which of these two logarithms is bigger: $8eds \geq \frac{m}{k}$ and $8eds < \frac{m}{k}$.
		The first case is easy to analyze, since if $8eds \geq \frac{m}{k}$, then we directly get $m \leq 2k \log_2 \mleft(8eds\mright)$ via substitution on the right-hand side.
		The other case leads to $m < 2k \log_2 \frac{m}{k}$, also via substitution on the right-hand side.
		Solving this with the Lambert $W$-function tells us that if $k > 0$ then $m < k e^{-W_{-1}\mleft(-\frac{\ln 2}{2}\mright)} = 4k$.
		Because $d \geq 1$ and $s \geq 1$ then $\log_2\mleft(8eds \mright) \geq \log_2 \mleft(4e \mright) > 2$, so $m < 4k < 2k \log_2 \mleft(8 eds\mright)$ for this other case as well.
	\end{proof}
	
	\subsection{Pseudo-Dimension for Quantum Circuits}
	
	\begin{proposition}[Stirling's approximation]\label{prop:stirling}
		\[\ln(n!) = n \ln n - n + \mathcal{O}(\ln n) = \mathcal{O}(n \ln n)\]
	\end{proposition}
	
	\begin{theorem}[Generalization of \citet{2020Caro} Theorem 3]\label{thm:Caro}
		The pseudo-dimension of quantum circuits on $n$ qudits comprised of at most $\gamma$ $2$-qudit gates with depth $\delta$ is upper bounded by $\mathcal{O}(d^4 \delta \gamma^2 \log \gamma)$.
	\end{theorem}
	\begin{proof}
		We want to apply \cref{thm:psuedo_dim_poly} to the concept class of quantum circuits. We know from \cref{cor:fixed_circuit_poly} that for fixed circuit structure with $\gamma$ gates and depth $\delta$ that it can be described as a polynomial with degree at most $2 \gamma$ in the $2 \gamma d^4$ real variables that describe the entries of the circuit. Furthermore, \cref{lemma:circuit_structures_num} tells us that there is at most $\frac{\gamma \delta^{\gamma - \delta}}{(\gamma - \delta)!}(n!)^\delta$ different circuit structures. We then apply \cref{thm:psuedo_dim_poly} with $k = 2 \gamma d^4$, $d = 2 \gamma$ and $s = \frac{\gamma! \delta^{\gamma - \delta}}{(\gamma - \delta)!}(n!)^\delta$.
		As a result we get that the pseudo-dimension is at most
		\begin{align}\label{eqn:pseudo_dim_apx}
			2k \log_2 \mleft(8 eds\mright) = 4 \gamma d^4 \log_2 \mleft(8e (2 \gamma) \frac{\gamma! \delta^{\gamma - \delta}}{(\gamma - \delta)!}(n!)^\delta\mright)
		\end{align}
		
		We will now focus on giving an upper bound for the logarithmic term by showing that
		\[
		\log_2 \mleft(8e (2 \gamma) \frac{\gamma! \delta^{\gamma - \delta}}{(\gamma - \delta)!}(n!)^\delta\mright) = \mathcal{O}(\delta \gamma \log \gamma).
		\]
		Splitting up the logarithm into sums and applying Stirling's approximation to each factorial, we arrive at
		\begin{align*}
			\log_2 \mleft(8e (2 \gamma) \frac{\gamma! \delta^{\gamma - \delta}}{(\gamma - \delta)!}(n!)^\delta\mright)
			=& \,\,4 + \log_2 e + \log_2 \gamma + (\gamma - \delta)\log_2 \gamma\\
			&+\mathcal{O}(\gamma \ln \gamma) + \delta \cdot \mathcal{O}(n \ln n) + \mathcal{O}((\delta - \gamma) \ln (\gamma - \delta)\\
			=& \,\,\mathcal{O}\mleft(\gamma \log \gamma + \delta n \log n + \delta \log \gamma \mright)
		\end{align*}
		
		Due to the definition of circuit structure, we know that $\delta \leq \gamma$.
		WLOG, we can also assume that every qubit has been acted upon by at least one gate (even if it's just the identity gate) such that $n \leq \gamma$.
		Together, we arrive that the logarithmic term is at most $\mathcal{O}(\delta \gamma \log \gamma)$.
		
		Since we have achieved our goal of upper-bounding the logarithmic term, from \cref{eqn:pseudo_dim_apx} we immediately get that the pseudo-dimension is at most $\mathcal{O}(d^4 \delta \gamma^2 \log \gamma)$.
	\end{proof}
	
	We now state the generalization of the main result \citet{2020Caro} to projective measurements of arbitrary rank.
	
	\begin{corollary}[Generalization of \citet{2020Caro} Corollary 3]\label{cor:Caro}
		Let $X$ be the set of quantum states on $n$ qudits, and let $Y$ be the set of all projectors on $n$ qudits. Let $U_*$ be a quantum circuit of 2-qudit quantum gates with size $\Gamma$ and depth $\Delta$. Let $\mathcal{D}$ be a probability distribution on $X \times Y$ unknown to the learner. Let
		\[S = \mleft\{\mleft(\mleft(x^{(i)}, y^{(i)}\mright), \trc\mleft[y^{(i)}U_* x^{(i)} U_*^\dagger\mright]\mright)\mright\}^m_{i=1}\]
		be corresponding training data where each $\mleft(x^{(i)}, y^{(i)}\mright)$ is drawn i.i.d according to $\mathcal{D}$. Let $\delta, \epsilon, \alpha, \beta \in (0, 1)$ where $\beta > \alpha$. Then, training data of size 
		\[m = \mathcal{O}\mleft(\frac{1}{\epsilon}\mleft(\Delta d^4\Gamma^2\log\Delta \log^2\mleft(\frac{\Delta d^4 \Gamma^2 \log(\Gamma)}{(\beta - \alpha)\epsilon}\mright) + \log\frac{1}{\delta}\mright)\mright)\]
		suffice to guarantee that, with probability $\geq 1- \delta$ with regard to choice of the training data, any quantum circuit $U$ of size $\Gamma$ and depth $\Delta$ that satisfies
		\begin{align*}
			\mleft|\trc\mleft[y^{(i)}U_* x^{(i)} U_*^\dagger\mright]  - \trc\mleft[y^{(i)}U x^{(i)} U^\dagger\mright] \mright| \leq \alpha \,\,\,\,\, \forall1 \leq i \leq m
		\end{align*}
		also satisfies
		\begin{align*}
			\ex_{(x, y) \sim \mathcal{D}}\mleft[\mleft(\trc\mleft[y^{(i)}U_* x^{(i)} U_*^\dagger\mright]  - \trc\mleft[y^{(i)}U x^{(i)} U^\dagger\mright]\mright)^2\mright] \leq (1 - \epsilon)\beta^2 + \epsilon
		\end{align*}
	\end{corollary}
	\begin{proof}
		We combine \cref{thm:Caro} with \cref{thm:fat_to_samp}, along with the fact that for all $\eta > 0$ the $\eta$-fat-shattering dimension is upper-bounded by pseudo-dimension.
	\end{proof}
	
	As shown by Theorem 4 of \citet{2020Caro}, a similar thing can be done with $n$-qudit quantum processes by simply changing the $d^4$ to $d^8$ in \cref{lem:fixed_circuit_poly,cor:fixed_circuit_poly}. This is because a quantum process is still a linear operation, but contains $d^8$ many entries now in parameter space. This propagates to \cref{thm:Caro,cor:Caro} by again replacing every appearance of $d^4$ with $d^8$.

	\section{Special Cases with Efficient Proper Learners}
	Despite the results given, there still exist situations where it is possible to efficiently proper learn Clifford circuits and CNOT circuits. We give brief proof sketches of some of them here.
	
	\subsection{CNOT Circuits for a Distribution with Support over a Single Measurement}\label{ssec:single_measurement}
	Let us try to learn CNOT circuits with regard to a distribution $\mathcal{D}$ such that there exists some pauli $P \in \{I, Z\}^{\otimes n}$ with $\pr_{(\rho, E) \sim \mathcal{D}}[E = \frac{I^{\otimes n}+P}{2}] = 1$. Because we are dealing with CNOT circuit the labels will always be $0$ and $1$ so by \cref{lemma:pauli_measure_trace} each label will tell us an affine subspace that $C^\dagger P C$ lies in. We can efficiently compute the intersection of this using Gaussian elimination with the generators to find a $P'$ that is consistent will all of the labels. From there, let $P, Q_2, Q_3, \dots, Q_n$ be a set of Paulis whose span is $\{I, Z\}^{\otimes n}$. Let $P', Q_2', Q_3', \dots, Q_n'$ also be a set of Paulis whose span is $\{I, Z\}^{\otimes n}$. It is clear that if we define our CNOT circuit such that $C^\dagger P C = P'$ and $C^\dagger Q_i C = Q_i'$ then we have a valid CNOT circuit. Efficiently finding such $\{Q_i\}$ and $\{Q_i'\}$ only takes $\mathcal{O}(n)$ expected samples of random Paulis in $\{I, Z\}^{\otimes n}$ and so can be done efficiently. Appealing to both \cref{prop:clifford_circuit_number} and \cref{thm:fat_to_samp} completes the proof.
	
	\subsection{Clifford circuits with the Uniform Distribution over Pauli Measurements}
	We note that if the distribution $\mathcal{D}$ entails the measurements being uniform over the Paulis then the problem is trivially easy to properly learn with $\epsilon < \frac{1}{\text{exp}(n)}$ and $\delta=0$ by just outputting a random Clifford circuit. This is because the probability that a random Pauli is in a given state's stabilizer group is $\frac{2^n}{4^n}$ so we will almost always see the label $\frac{1}{2}$ regardless of the hypothesis circuit we choose.
	
	\subsection{CNOT Circuits with the Uniform Distribution over \texorpdfstring{$\{I, Z\}^{\otimes n}$}{\{I, Z\}⊗n}}
	Let $\Theta$ and $Q$ be the matrix/vector forms of $\theta_{ij}$ and $q_j$ values from \cref{eq:clifford_gen}. We note that if we have enough independent samples that $\Theta \oplus Q$ is confined to a $\mathcal{O}(\log n)$ dimensional affine subspace then we can simply iterate through all possible $\Theta$ and $Q$ combinations to find one with a non-singular $\Theta$ in $\text{poly}(n)$ time. Let's say that we've restricted $\Theta \oplus Q$ to lie in a $d$ dimensional affine subspace. Let $M'$ be the true value of $\Theta$. and let $M \neq M'$ be another arbitrary matrix. Likewise let $x'$ be the true value of $Q$ and $x$ some arbitrary vector. Now let $Z^w \in \{I, Z\}^{\otimes n}$ be a Pauli selected uniformly at random, and $\rho = \ketbra{v}{v}$ a uniformly random computational basis state. The pair $M$ and $x$ will give the same label as the true label on the input $(\rho, \frac{I^{\otimes n} + Z^w}{2})$ if and only if
	\[
	\mleft(v^T(M + M') + x^T + (x')^T\mright)w = 0 \mod 2.
	\]
	Because $w$ is uniform random, as long as $v^T(M + M') + x^T + (x')^T$ is not the zero vector over $\mathbb{F}_2$ then this will only be $0$ at most half the time. Since at least one of $M \neq M'$ or $x \neq x'$ is true, then $(M+M')^T v = x + x'$ will only be true with probability at most $\frac{1}{2}$ as well. So the probability that any arbitrary $M$ and $x$ have different labels is at least $\frac{1}{4}$. Thus with $\mathcal{O}(n)$ expected samples uniformly drawn from arbitrary basis states and $Z^w$ we will have constrained our system to something we can bruteforce to find a full rank $\Theta$ and corresponding $q_j$ values that is consistent with all samples. From there we again apply both \cref{prop:clifford_circuit_number} and \cref{thm:fat_to_samp} to generalize with zero training error as long as the number of samples is also at least the parameter $m$ from the theorem statement.

	\subsection{Clifford Circuits for a Distribution with Support over a Single State}
	In the converse of an earlier situation, let us try to learn Clifford circuits with regard to a distribution $\mathcal{D}$ such that there exists some stabilizer state $\sigma$ with $\pr_{(\rho, E) \sim \mathcal{D}}[\rho = \sigma] = 1$. This situation effectively reduces to that of \citet{rocchetto2018stabiliser}. If we run that algorithm we will find a state $\sigma'$ that is consistent with all of the labels. Let $\{g_i\}$ be the generators of $\sigma$ and $\{g_i'\}$ the generators of $\sigma'$. If we let $C g_i C^\dagger = g_i'$ we define the first part of a Clifford circuit that maps $\sigma$ to $\sigma'$ as desired. We can then run the algorithm from \citet{https://doi.org/10.48550/arxiv.2008.06011} to fill in the remainder of the Clifford circuit. Appealing to both \cref{prop:clifford_circuit_number} and \cref{thm:fat_to_samp} once again completes the proof.

\end{document}